\documentclass[twocolumn, a4paper, accepted=2020-05-20]{quantumarticle}
\pdfoutput=1

\usepackage{array}
\usepackage{graphicx}
\usepackage{dcolumn}
\usepackage{bm}
\usepackage{amsmath}
\usepackage{amsfonts}
\usepackage{amssymb}
\usepackage{amstext}   
\usepackage{amsthm}
\usepackage{makecell}
\usepackage{bm}
\usepackage{dsfont}
\usepackage{nameref}
\usepackage[T1]{fontenc}
\usepackage{float}
\usepackage[utf8]{inputenc}
\usepackage[outdir=./]{epstopdf}
\usepackage{subfloat}
\usepackage[colorlinks = true, citecolor = blue]{hyperref}
\usepackage{multirow}
\usepackage{manfnt}
\usepackage{makecell}
\usepackage{xcolor}
\usepackage{mathtools}
\usepackage{mathdots}
\usepackage{wasysym}
\usepackage{xfrac}
\newcolumntype{C}{>{$}c<{$}}
\AtBeginDocument{
\heavyrulewidth=.08em
\lightrulewidth=.05em
\cmidrulewidth=.03em
\belowrulesep=.65ex
\belowbottomsep=0pt
\aboverulesep=.4ex
\abovetopsep=0pt
\cmidrulesep=\doublerulesep
\cmidrulekern=.5em
\defaultaddspace=.5em
\tabcolsep=7pt
}
\usepackage{booktabs}
\usepackage[numbers,sort&compress]{natbib}
\usepackage{enumerate}

\definecolor{emerald}{rgb}{0.07, 0.53, 0.03}
\definecolor{babyblueeyes}{rgb}{0.63, 0.79, 0.95}

\newcommand{\ket}[1]{\ensuremath{\left|#1\right\rangle}}

\newcommand{\scomm}[2]{\ensuremath{[\![ #1, #2]\!]}}

\newcommand{\myfancysymbol}{
 {\mathchoice
  {\includegraphics[height=1.6ex]{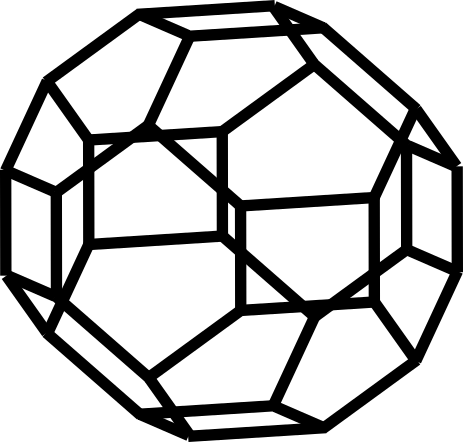}}
  {\includegraphics[height=1.6ex]{BallCell.png}}
  {\includegraphics[height=1.2ex]{BallCell.png}}
  {\includegraphics[height=0.9ex]{BallCell.png}}
 }
}

\newtheorem{theorem}{Theorem}
\newtheorem{lemma}{Lemma}
\newtheorem{corollary}{Corollary}[theorem]

\newtheorem{definition}{Definition}

\newtheorem{innercustomresult}{Result}
\newenvironment{customresult}[1]
  {\renewcommand\theinnercustomresult{#1}\innercustomresult}
  {\endinnercustomresult}
  
  \newtheorem{innercustomthm}{Theorem}
\newenvironment{customthm}[1]
  {\renewcommand\theinnercustomthm{#1}\innercustomthm}
  {\endinnercustomthm}
  
\newtheorem{innercustomcor}{Corollary}
\newenvironment{customcor}[1]
  {\renewcommand\theinnercustomcor{#1}\innercustomcor}
  {\endinnercustomcor}

\hypersetup{linkcolor=blue,urlcolor=blue}

\begin{document}

\title{Characterization of solvable spin models via graph invariants}
\author[1]{Adrian Chapman}
\email{adrian.chapman@sydney.edu.au}
\author[1]{Steven T. Flammia}
\orcid{0000-0002-3975-0226}
\affil[1]{Centre for Engineered Quantum Systems, School of Physics, The University of Sydney, Sydney, Australia} 

\date{May 27, 2020}

\begin{abstract}
Exactly solvable models are essential in physics. 
For many-body spin-$\mathbf{\sfrac{1}{2}}$ systems, an important class of such models consists of those that can be mapped to free fermions hopping on a graph. 
We provide a complete characterization of models which can be solved this way. 
Specifically, we reduce the problem of recognizing such spin models to the graph-theoretic problem of recognizing line graphs, which has been solved optimally. 
A corollary of our result is a complete set of constant-sized commutation structures that constitute the obstructions to a free-fermion solution.  
We find that symmetries are tightly constrained in these models. 
Pauli symmetries correspond to either: (i) cycles on the fermion hopping graph, (ii) the fermion parity operator, or (iii) logically encoded qubits. 
Clifford symmetries within one of these symmetry sectors, with three exceptions, must be symmetries of the free-fermion model itself. 
We demonstrate how several exact free-fermion solutions from the literature fit into our formalism and give an explicit example of a new model previously unknown to be solvable by free fermions.
\end{abstract}

\maketitle

\section{Introduction}

Exactly solvable models provide fundamental insight into physics without the need for difficult numerical methods or perturbation theory.
In the particular setting of many-body spin-$\sfrac{1}{2}$ systems, a remarkable method for producing exact solutions involves finding an effective description of the system by noninteracting fermions. 
This reduces the problem of solving the $n$-spin system over its full $2^n$-dimensional Hilbert space to one of solving a single-particle system hopping on a lattice of $O(n)$ sites. 
The paradigmatic example of this method is the exact solution for the XY model \cite{lieb1961two}, where the Jordan-Wigner transformation \cite{jordan1928uber} is employed to describe the model in terms of free fermions propagating in one spatial dimension. 
These fermions are resolved as nonlocal Pauli operators in the spin picture, and the nonlocal nature of this mapping may suggest that finding generalizations to this mapping for more complicated spin systems is a daunting task. 
Of the many generalizations that have since been proposed \cite{fradkin1989jordanwigner, wang1991ground, huerta1993bosefermi, batista2001generalized, verstraete2005mapping, nussinov2012arbitrary, chen2018exact, backens2019jordanwigner, tantivasadakarn2020jordanwigner}, a particularly interesting solution to this problem is demonstrated in the exact solution of a 2-d spin model on a honeycomb lattice introduced by Kitaev \cite{kitaev2006anyons}. For this model, the transformation to free-fermions can be made locality-preserving over a fixed subspace through the use of local symmetries.

The dynamics of free-fermion systems are generated by Gaussian-fermionic Hamiltonians and correspond to the class of so-called matchgate circuits. 
This circuit class coincides with the group of free-fermion propagators generated by arbitrarily time-dependent single-particle Hamiltonians \cite{knill2001fermionic, terhal2002classical} and has an extensive complexity-theoretic characterization. 
In general, matchgate circuits can be efficiently simulated classically with arbitrary single-qubit product-state inputs and measurement \cite{vandennest2011simulating, brod2016efficient}. 
However, they become universal for quantum computation with the introduction of non-matchgates such as the $\mathrm{SWAP}$ gate \cite{jozsa2008matchgates, brod2011extending}, certain measurements and resource inputs \cite{bravyi2006universal, hebenstreit2019all}, and when acting on nontrivial circuit geometries \cite{brod2014computational}. 
Furthermore, these circuits share an interesting connection to the problem of counting the number of perfect matchings in a graph, which is the context in which they were first developed \cite{valiant2002quantum, cai2006valiants, cai2007theory, valiant2008holographic}. 
This problem is known to be very hard computationally (it is \#P-complete \cite{papadimitriou1994computational}), but is efficiently solved for planar graphs using the so-called Fisher-Kasteleyn-Temperley algorithm \cite{kasteleyn1961statistics, temperley1961dimer}.

In this work, we develop a distinct connection between free-fermion systems and graph theory by using tools from quantum information science. 
The central object of our formalism is the \emph{frustration graph}. 
This is a network quantifying the anticommutation structure of terms in the spin Hamiltonian when it is expanded in the basis of Pauli operators \cite{planat2008pauli}. 
This graph has been invoked previously in the setting of variational quantum eigensolvers \cite{jena2019pauli, verteletskyi2019measurement, zhao2019measurement, izmaylov2019unitary, yen2019measuring, gokhale2019minimizing, crawford2019efficient, bonetmonroig2019nearly}, commonly under the name ``anti-compatibility graph". 
We show that the problem of recognizing whether a given spin model admits a free-fermion solution is equivalent to that of recognizing whether its frustration graph is a \emph{line graph}, which can be performed optimally in linear time \cite{roussopoulos1973max, lehot1974optimal, degiorgi1995dynamic}. 
From the definition of a line graph, it will be clear that such a condition is necessary, but we will show that it is also sufficient. 
When the condition is met, we provide an explicit solution to the model. 

Line graphs have recently emerged as the natural structures describing the effective tight-binding models for superconducting waveguide networks \cite{Kollar2019Hyperbolic, Kollar2019Linegraph,Boettcher2019Quantum}. 
In this setting, the line graph corresponds to the physical hopping graph of photons in the network. 
We will see how this scenario is a kind of ``inverse problem" to the one we consider, wherein fermions are hopping on the \emph{root} of the line graph. 
It is clear from both scenarios that the topological connectivity structure of many-body systems plays a central role in their behavior, and it is remarkable that this is already being observed in experiments. 
We expect that further investigation of the graph structure of many-body Hamiltonians will continue to yield important insights into their physics.

\subsection{Summary of Main Results}
Here we give a brief summary of the main results. 
We first define the frustration graph of a Hamiltonian, given in the Pauli basis, as the graph with nonzero Pauli terms as vertices and an edge between two vertices if their corresponding terms anticommute. 
A line graph $G$ of a graph $R$ is the intersection graph of the edges of $R$ as two-element subsets of the vertices of $R$. 
With these simple definitions, we can informally state our first main result, which we call our ``fundamental theorem:"

\begin{customresult}{1}[Existence of free-fermion solution; Informal version of Thm.~\ref{thm:ffsolution}] Given an $n$-qubit Hamiltonian in the Pauli basis for which the frustration graph $G$ is the line graph of another graph $R$, then there exists a free-fermion description of $H$. 
\label{res:ffsolution}
\end{customresult}

From this description, an exact solution for the spectrum and eigenstates of $H$ can be constructed. 
This theorem illustrates a novel connection between the physics of quantum many-body systems and graph theory with some surprising implications. 
First, it gives the exact correspondence between the spatial structure of a spin Hamiltonian and that of its effective free-fermion description. 
As we will see through several examples, this relationship is not guaranteed to be straightforward. 
Second, the theorem gives an exact condition by which a spin model can \emph{fail} to have a free-fermion solution, the culprit being the presence of forbidden anticommutation structures in the frustration graph of $H$.

Some caveats to Result~\ref{res:ffsolution} (that are given precisely in the formal statement, Theorem \ref{thm:ffsolution}) involve cases in which this mapping between Pauli terms in $H$ and fermion hopping terms is not one-to-one. 
In particular, if we are given a Hamiltonian whose frustration graph is not a line graph, then a free-fermion solution may still be possible via a non-injective mapping over a subspace defined by fixing stabilizer degrees of freedom. 
Additionally, it is possible for a given spin Hamiltonian to describe multiple free-fermion models simultaneously, each generating dynamics over an independent stabilizer subspace of the full Hilbert space as for the Kitaev honeycomb model \cite{kitaev2006anyons}. 
These symmetries are sometimes referred to as gauge degrees of freedom, though we will reserve this term for freedoms which cannot affect the physics of the free-fermion model.
Finally, it may be the case that the free-fermion model contains states which are nonphysical in the spin-Hamiltonian picture, and so these must be removed by fixing a symmetry as well. 
Luckily, all of these cases manifest as structures in the frustration graph of $H$. 
The first, regarding when a non-injective free-fermion solution is required, is signified by the presence of so-called twin vertices, or vertices with the same neighborhood. 
We deal with this case in our first lemma. 
The next two cases are covered by our second theorem:

\begin{customresult}{2}[Graphical symmetries; Informal version of Thm.~\ref{thm:symmetries}] Given an $n$-qubit Hamiltonian in the Pauli basis for which the frustration graph $G$ is the line graph of another graph $R$, then Pauli symmetries of $H$ correspond to either:
\begin{enumerate}[(i)]
  \item Cycles of $R$;
  \item A T-join of $R$, associated to the fermion-parity operator;
  \item Logically encoded qubits;
\end{enumerate}
and these symmetries generate an abelian group.
\end{customresult}

\noindent We then prove that we can always fix all of the cycle symmetries by choosing an orientation of the root graph $R$. 
Our results also relate the more general class of Clifford symmetries to the symmetries of the single-particle free-fermion Hamiltonian. 
We show that with exactly three exceptions, Clifford symmetries of the spin model, in a subspace defined by fixing the symmetries listed above, must also be symmetries of the single-particle Hamiltonian (see Corollary~\ref{cor:clifsym} for a precise statement).

Finally, we illustrate these ideas with several examples: small systems of up to 3 qubits, the 1-dimensional anisotropic $XY$ model in a transverse field and its nearest-neighbor solvable generalization, the Kitaev honeycomb model, the 3-dimensional frustrated hexagonal gauge color code \cite{roberts2019symmetry}, and the Sierpinski-Hanoi model. 
To the best of our knowledge, this last model was previously not known to be solvable.

The remainder of the paper is organized as follows. 
In Section \ref{sec:bkgrnd}, we will introduce notation and give some background on the formalism of free-fermions and frustration graphs. 
In Section \ref{sec:fundthm}, we will formally state Theorem 1 and some general implications thereof. 
In Section \ref{sec:symmetries}, we elaborate on the structure of symmetries which can be present in our class of solvable models. 
In Section \ref{sec:orientation}, we will use the theorems of the previous two sections to outline an explicit solution method. 
We close by demonstrating how the examples of free-fermion solutions listed above fit into this formalism in Section \ref{sec:examples}.

\section{Background}
\label{sec:bkgrnd}
\subsection{Frustration Graphs}

The models we consider are spin-$\sfrac{1}{2}$ (qubit) Hamiltonians written in the Pauli basis
\begin{align}
    H = \sum_{\boldsymbol{j} \in V} h_{\boldsymbol{j}} \sigma^{\boldsymbol{j}} \mathrm{,}
\label{eq:hdef}
\end{align}
where $\boldsymbol{j} \equiv (\boldsymbol{a}, \boldsymbol{b})$, with $\boldsymbol{a}$, $\boldsymbol{b} \in \{0, 1\}^{\times n}$ labeling an $n$-qubit Pauli operator as
\begin{align}
    \sigma^{\boldsymbol{j}} = i^{\boldsymbol{a} \cdot \boldsymbol{b}} \left(\bigotimes_{k = 1}^n X_k^{a_k} \right)\left(\bigotimes_{k = 1}^n Z_k^{b_k} \right) \mathrm{.}
\end{align}
The exponent of the phase factor, $\boldsymbol{a} \cdot \boldsymbol{b}$, is the \emph{Euclidean inner product} between $\boldsymbol{a}$ and $\boldsymbol{b}$. 
This phase is chosen such that the overall operator is Hermitian, and such that $a_k = b_k = 1$ means that $\sigma^{\boldsymbol{j}}$ acts on qubit $k$ by a Pauli-$Y$ operator. 
We denote the full $n$-qubit Pauli group by $\mathcal{P}$, and $V \subseteq \mathcal{P}$ is the set of Pauli terms in $H$ (i.e.\ $h_{\boldsymbol{j}} = 0$ for all $\boldsymbol{j} \notin V$). 
Let the Pauli subgroup generated by this set be denoted $\mathcal{P}_H$. 

For our purposes, what is important is not the explicit Pauli description of the Hamiltonian, but rather the commutation relations between its terms. 
As Pauli operators only either commute or anticommute, a useful quantity is their \emph{scalar commutator} $\scomm{\cdot}{\cdot}$, which we define implicitly as
\begin{align}
    \sigma^{\boldsymbol{j}} \sigma^{\boldsymbol{k}} = \scomm{\sigma^{\boldsymbol{j}}}{\sigma^{\boldsymbol{k}}} \sigma^{\boldsymbol{k}} \sigma^{\boldsymbol{j}} \mathrm{.}
\end{align}
The scalar commutator thus only takes the values $\pm 1$. 
Additionally, the scalar commutator distributes over multiplication in each argument, e.g.
\begin{align}
    \scomm{\sigma^{\boldsymbol{j}}}{\sigma^{\boldsymbol{k}}\sigma^{\boldsymbol{l}}} = \scomm{\sigma^{\boldsymbol{j}}}{\sigma^{\boldsymbol{k}}} \scomm{\sigma^{\boldsymbol{j}}}{\sigma^{\boldsymbol{l}}}.
\label{eq:distribution}
\end{align}
For $n$-qubit Paulis, the scalar commutator can thus be read off from the Pauli labels as
\begin{align}
    \scomm{\sigma^{\boldsymbol{j}}}{\sigma^{\boldsymbol{k}}} = (-1)^{\langle \boldsymbol{j}, \boldsymbol{k} \rangle}
\label{eq:symplscomm}
\end{align}
Here, $\langle \boldsymbol{j}, \boldsymbol{k} \rangle$ is the \emph{symplectic inner product}
\begin{align}
    \langle \boldsymbol{j}, \boldsymbol{k} \rangle \equiv \begin{pmatrix} \boldsymbol{a}_j & \boldsymbol{b}_j \end{pmatrix} \begin{pmatrix} \mathbf{0}_{n} & \mathbf{I}_n \\ 
    - \mathbf{I}_n & \mathbf{0}_{n} \end{pmatrix} \begin{pmatrix} \boldsymbol{a}_k \\ \boldsymbol{b}_k \end{pmatrix} \mathrm{,}
\label{eq:sympldef}
\end{align} 
where naturally $\boldsymbol{j} \equiv (\boldsymbol{a}_j, \boldsymbol{b}_j)$ and $\boldsymbol{k} \equiv (\boldsymbol{a}_k, \boldsymbol{b}_k)$. 
$\mathbf{0}_{n}$ is the $n \times n$ all-zeros matrix, and $\mathbf{I}_n$ is the $n \times n$ identity matrix. 
Eq.~(\ref{eq:symplscomm}) captures the fact that a factor of $-1$ is included in the scalar commutator for each qubit where the operators $\sigma^{\boldsymbol{j}}$ and $\sigma^{\boldsymbol{k}}$ differ and neither acts trivially. 
Since the inner product appears as the exponent of a sign factor, without loss of generality, we can replace it with the \emph{binary symplectic inner product}
\begin{align}
    \langle \boldsymbol{j}, \boldsymbol{k} \rangle_2 \equiv \langle \boldsymbol{j}, \boldsymbol{k} \rangle \bmod 2.
\end{align}

\begin{table}[t]
\centering
\setcellgapes{3pt}
\makegapedcells
\begin{tabular}{c c c}
\toprule
$H$ & $\sum\limits_{j \in \{x, y, z\}} h_j\sigma^{j}$ & $\sum\limits_{\substack{\boldsymbol{j} \in \{0, x, y, z\}^{\times 2}\\ \boldsymbol{j} \neq (0, 0)}} h_{\boldsymbol{j}}\sigma^{\boldsymbol{j}}$ \\ 
\midrule 
\makecell{$G(H)$ \\ $\simeq L(R)$} & \makecell{\includegraphics[width=0.1\textwidth]{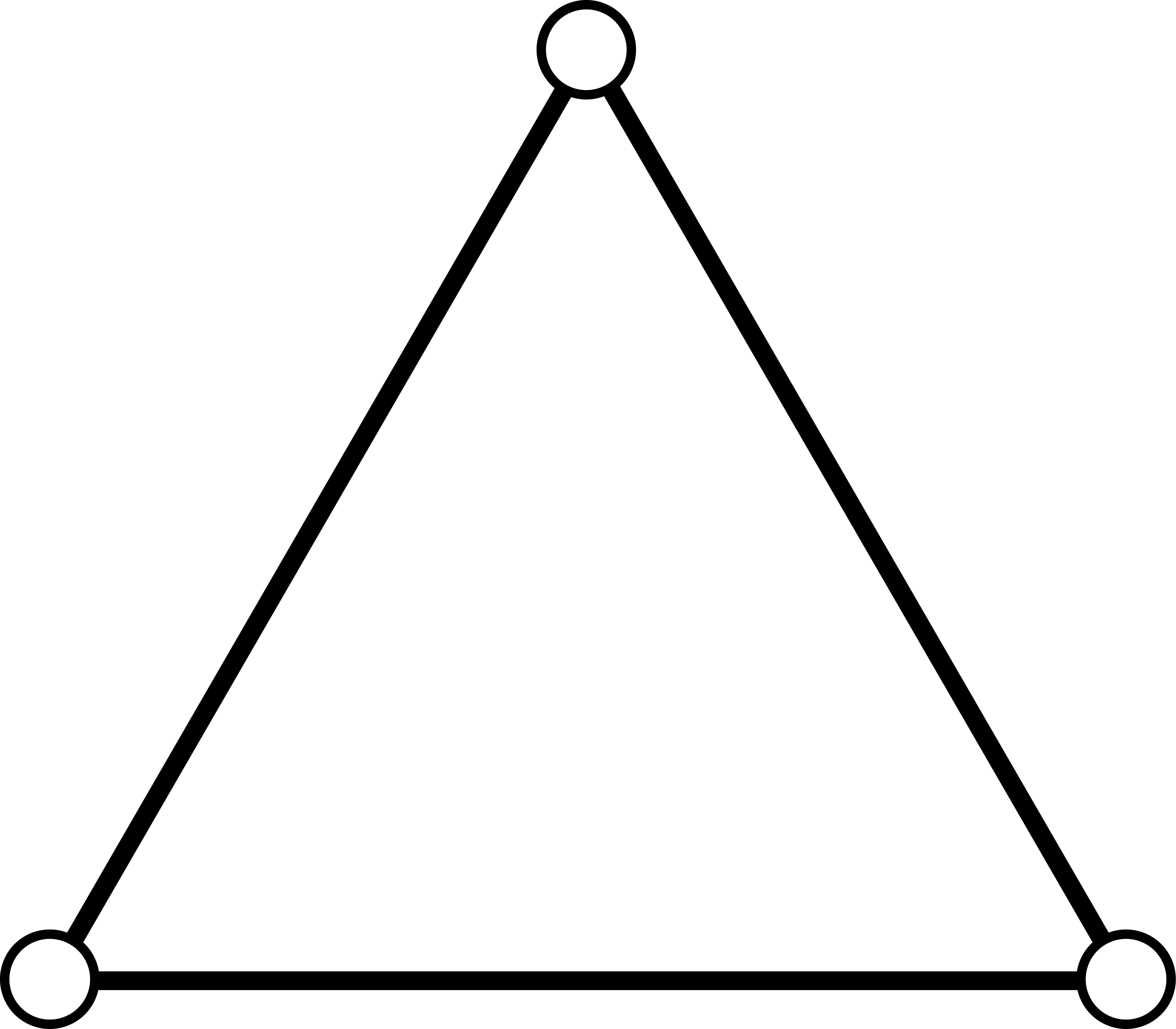}} &  \makecell{\includegraphics[width=0.15\textwidth]{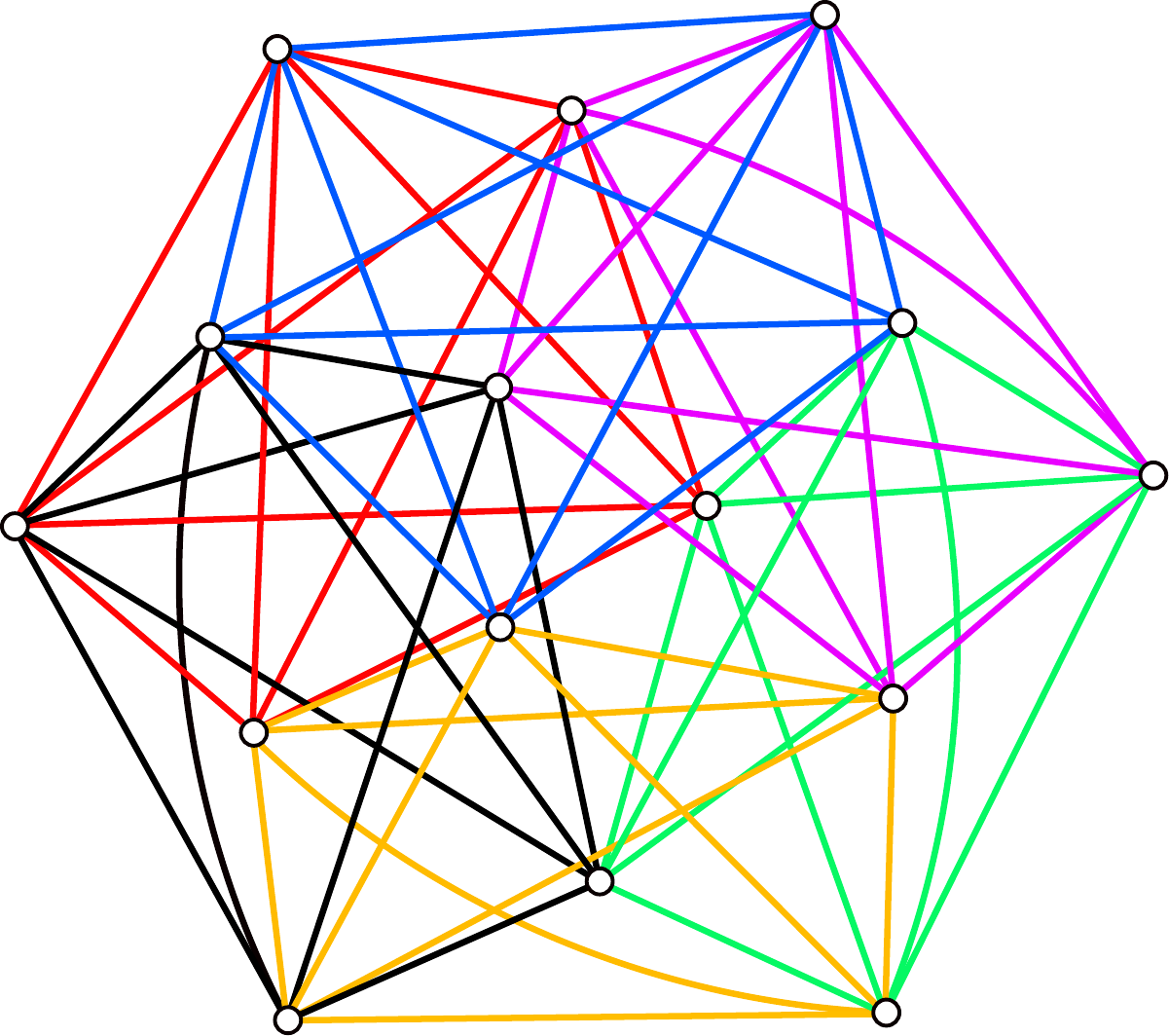}} \\ \cline{1-3}
$R$ & \makecell{\includegraphics[width=0.07\textwidth]{K3Graph}} or \makecell{\includegraphics[width=0.07\textwidth]{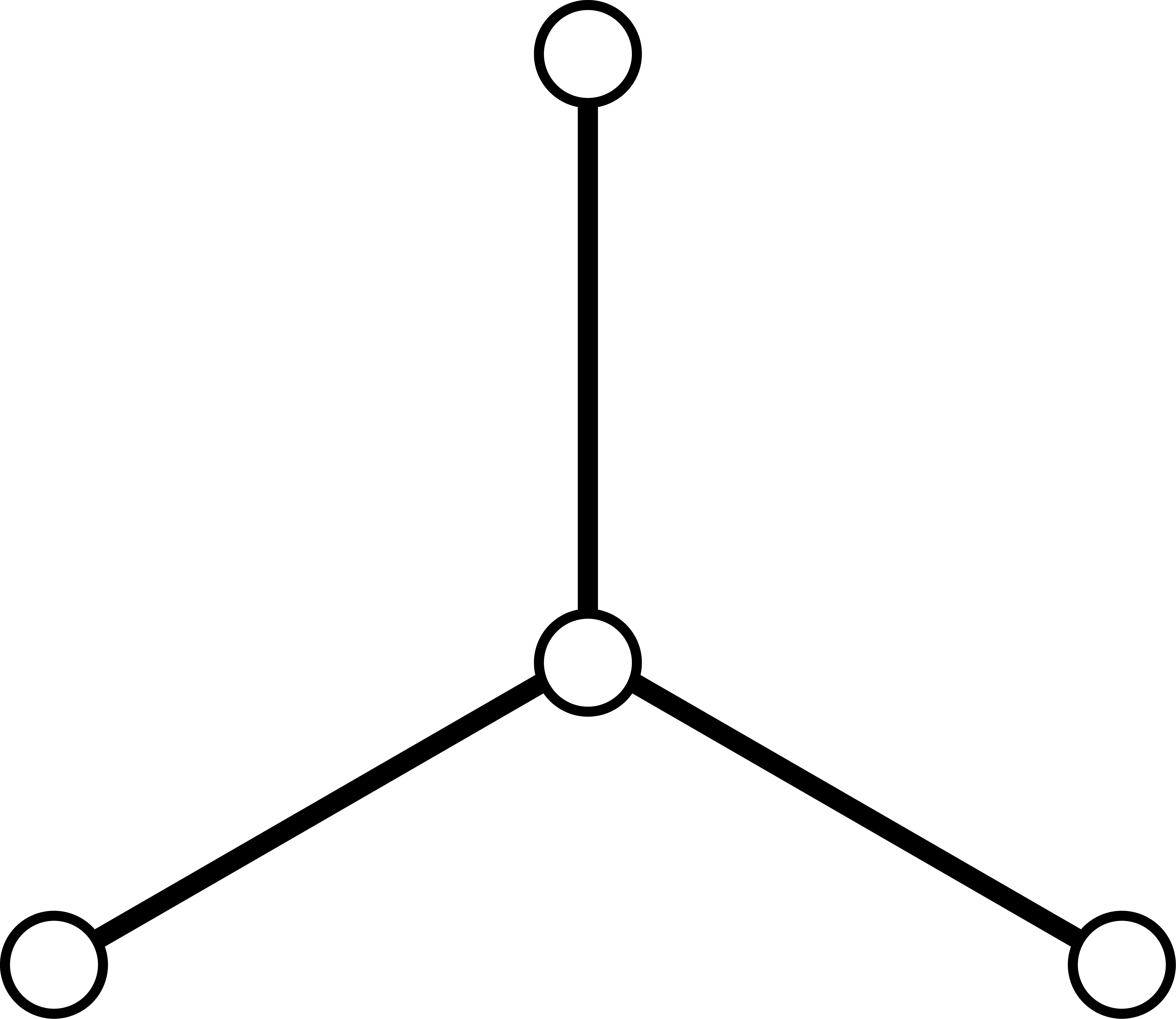}} & \makecell{\includegraphics[width=0.1\textwidth]{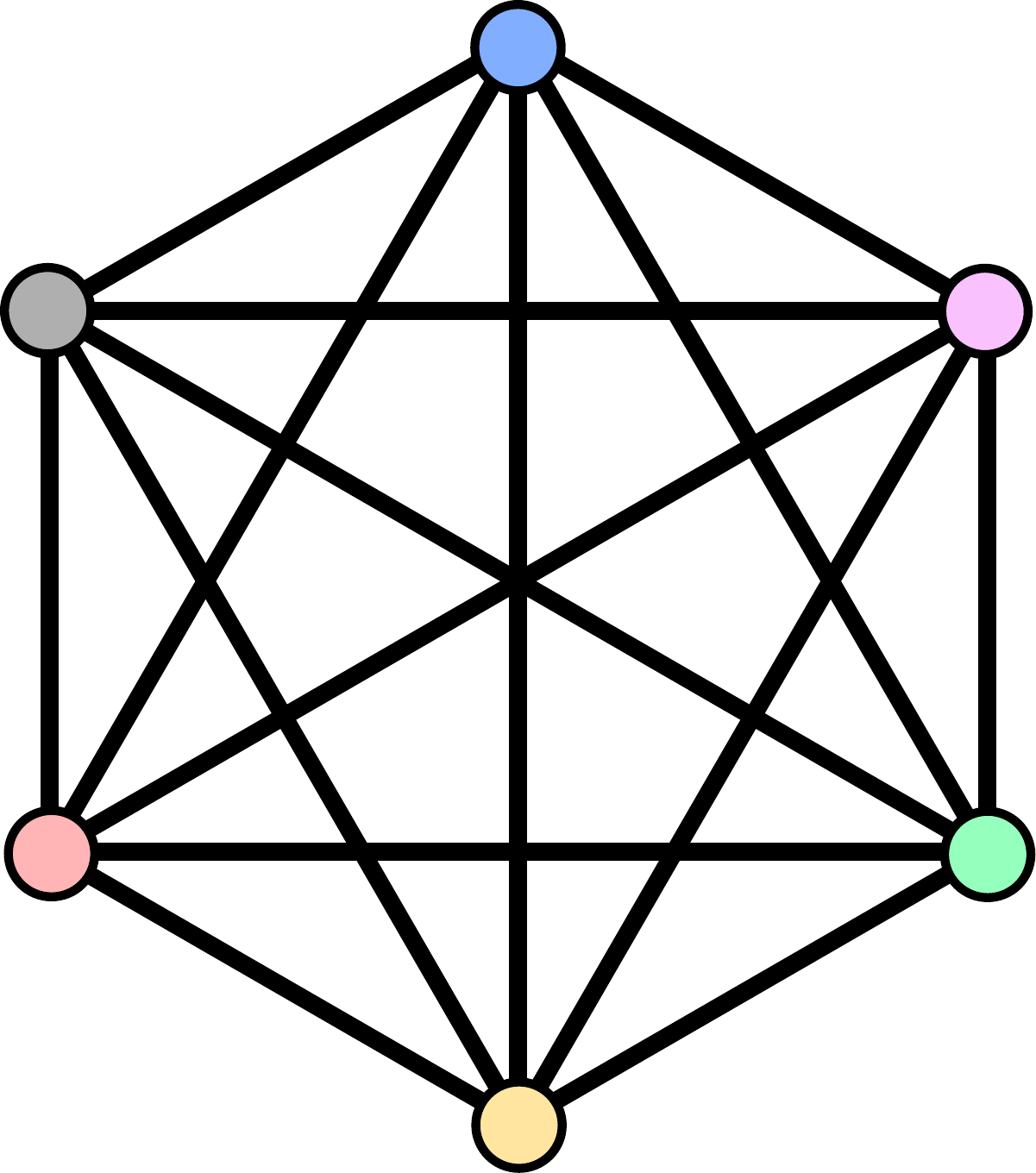}} \\
\bottomrule
\end{tabular}
\caption{Example frustration graphs for general Hamiltonians on small (1- and 2-qubit) systems. 
(Left column) For general single-qubit Hamiltonians, the frustration graph is the complete graph on three vertices, $K_3$. 
By the Whitney isomorphism theorem \cite{whitney1932congruent}, $K_3$ is the only graph which is not the line graph of a unique graph, but rather is the line graph of both $K_3$ and the `claw' graph, $K_{1, 3}$. 
This implies the existence of two distinct free-fermion solutions of single-qubit Hamiltonians. 
(Right column) For general two-qubit Hamiltonians, the frustration graph is the line graph of the complete graph on six vertices $K_6$ \cite{, planat2008pauli, goodmanson1996graphical}. 
Colored are the size-five cliques corresponding to the degree-five vertices of the root graph. 
This mapping implies the existence of a free-fermion solution for general two-qubit Hamiltonians by six fermions, reflecting the accidental Lie-algebra isomorphism $\mathfrak{su}(4) \simeq \mathfrak{spin}(6)$ (see Section \ref{sec:small}).}
\label{tab:notsum}
\end{table}

Through the binary symplectic inner product, the scalar commutator defines a symmetric binary relation between terms in the Hamiltonian, to which we associate the adjacency matrix of a graph. 
Denote the \emph{frustration graph} for a Hamiltonian of the form in Eq.~(\ref{eq:hdef}) by $G(H) \equiv (V, E)$ with vertex set given by the Pauli terms appearing in $H$, and edge set
\begin{align}
    E \equiv \{(\boldsymbol{j}, \boldsymbol{k}) | \langle \boldsymbol{j}, \boldsymbol{k} \rangle_2 = 1\}
\end{align}
That is, two Pauli terms correspond to neighboring vertices in $G(H)$ if and only if they anticommute. 
Without loss of generality, we can assume that $G(H)$ is connected, as disconnected components of this graph correspond to commuting collections of terms in the Hamiltonian and can thus be independently treated. 
As such, we will further assume that $H$ has no identity component in the expansion (\ref{eq:hdef})---rendering it traceless---since this will only contribute an overall energy shift to the system with no effect on dynamics.

\subsection{Majorana Fermions}
A related set of Hermitian operators which only either commute or anticommute is that of the Majorana fermion modes $\{\gamma_{\mu}\}_{\mu}$, which satisfy the canonical anticommutation relations
\begin{align}
    \gamma_{\mu} \gamma_{\nu} + \gamma_{\nu} \gamma_{\mu} = 2\delta_{\mu \nu} I \mathrm{,}
\label{eq:car}
\end{align}
and for which $\gamma_{\mu}^{\dagger} = \gamma_{\mu}$. 
A familiar way of realizing these operators in terms of $n$-qubit Pauli observables is through the Jordan-Wigner transformation
\begin{align}
    \gamma_{2j - 1} = \bigotimes_{k = 1}^{j - 1} Z_k \otimes X_j \mbox{\hspace{10mm}} \gamma_{2j} = \bigotimes_{k = 1}^{j - 1} Z_k \otimes Y_j \mathrm{.}
\label{eq:jwtransform}
\end{align}
The Pauli operators on the right can easily be verified to constitute $2n$ operators satisfying Eq.~(\ref{eq:car}). 
Of course, we will explore the full set of generalizations to this transformation in this work. 
We seek to identify those qubit Hamiltonians which can be expressed as quadratic in the Majorana modes. 
Such \emph{free-fermion} Hamiltonians are written as
\begin{align}
    \widetilde{H} =  i \boldsymbol{\gamma} \cdot \mathbf{h} \cdot \boldsymbol{\gamma}^{\mathrm{T}} \equiv 2i \sum_{(j, k) \in \widetilde{E}} h_{jk} \gamma_j \gamma_k
\label{eq:ffhamiltonian}
\end{align}
where $\boldsymbol{\gamma}$ is a row-vector of the Majorana operators, and $\mathbf{h}$ is the \emph{single-particle Hamiltonian}. 
Without loss of generality, $\mathbf{h}$ can be taken as a real antisymmetric matrix, as we can similarly assume $\widetilde{H}$ is traceless, and the canonical anticommutation relations Eq.~(\ref{eq:car}) guarantee that any symmetric component of $\mathbf{h}$ will not contribute to $\widetilde{H}$. 
$\widetilde{E}$ is the edge-set of the \emph{fermion-hopping graph} $R \equiv (\widetilde{V}, \widetilde{E})$ on the fermion modes $\widetilde{V}$. 
That is, $h_{jk} = 0$ for those pairs $(j, k) \notin \widetilde{E}$, and the factor of two in the rightmost expression accounts for the fact that each edge in $\widetilde{E}$ is included only once in the sum.

As a result of the canonical anticommutation relations (\ref{eq:car}), the individual Majorana modes transform covariantly under the time evolution generated by $\widetilde{H}$
\begin{align}
    \mathrm{e}^{i \widetilde{H} t} \gamma_{\mu} \mathrm{e}^{-i \widetilde{H} t} = \sum_{\nu \in \widetilde{V}} \left(\mathrm{e}^{4 \mathbf{h} t}\right)_{\mu \nu} \gamma_{\nu}  
\end{align}
since
\begin{align}
    [\bm{\gamma} \cdot \mathbf{h} \cdot \boldsymbol{\gamma}^{\mathrm{T}}, \gamma_{\mu}] = -4 (\mathbf{h} \cdot \boldsymbol{\gamma}^{\mathrm{T}})_{\mu}\,.
\end{align}
Since $\mathbf{h}$ is antisymmetric and real, $\mathrm{e}^{4\mathbf{h} t}\in \mathrm{SO}(2n, \mathds{R})$. 
Thus, $\mathbf{h}$ can be block-diagonalized via a real orthogonal matrix, $\mathbf{W}\in \mathrm{SO}(2n, \mathds{R})$, as
\begin{align}
    \mathbf{W}^{\mathrm{T}} \cdot \mathbf{h} \cdot \mathbf{W} = \bigoplus_{j = 1}^n \begin{pmatrix}
0 & -\lambda_j \\
\lambda_j & 0 \\
\end{pmatrix}
\end{align}
We can represent $\mathbf{W}$ as the exponential of a quadratic Majorana fermion operator as well, by defining 
\begin{align}
    \mathbf{W} \equiv \mathrm{e}^{4 \mathbf{w}} \mathrm{,}
\end{align}
$\widetilde{H}$ is therefore diagonalized as
\begin{align}
    \mathrm{e}^{- \boldsymbol{\gamma}\cdot \mathbf{w} \cdot \boldsymbol{\gamma}^{\mathrm{T}}} \widetilde{H} \mathrm{e}^{\boldsymbol{\gamma}\cdot \mathbf{w} \cdot \boldsymbol{\gamma}^{\mathrm{T}}} &= i \boldsymbol{\gamma} \cdot \left(\mathbf{W}^{\mathrm{T}} \cdot \mathbf{h} \cdot \mathbf{W} \right) \cdot \boldsymbol{\gamma}^{\mathrm{T}} \\
    &= -2i \sum_{j = 1}^n \lambda_j \gamma_{2j - 1} \gamma_{2j} \\ 
    \mathrm{e}^{- \boldsymbol{\gamma}\cdot \mathbf{w} \cdot \boldsymbol{\gamma}^{\mathrm{T}}} \widetilde{H} \mathrm{e}^{\boldsymbol{\gamma}\cdot \mathbf{w} \cdot \boldsymbol{\gamma}^{\mathrm{T}}} &= 2\sum_{j = 1}^{n} \lambda_j Z_j
\end{align}
Note that the exact diagonalization can be performed with reference to the \emph{quadratics} in the Majorana fermion modes only. 
To completely solve the system, it is only necessary to diagonalize $\mathbf{h}$ classically, find a generating matrix $\mathbf{w}$, and diagonalize $\widetilde{H}$ using an exponential of quadratics with regard to some fermionization like Eq.~(\ref{eq:jwtransform}). 
Eigenstates of $\widetilde{H}$ can be found by acting $\mathrm{e}^{\boldsymbol{\gamma}\cdot \mathbf{w} \cdot \boldsymbol{\gamma}^{\mathrm{T}}}$ on a computational basis state $\ket{\mathbf{x}}$ for $\mathbf{x} \in \{0, 1\}^{\times n}$. 
The associated eigenvalue is
\begin{align}
    E_{\mathbf{x}} = 2\sum_{j = 1}^n (-1)^{x_j} \lambda_j
\end{align}
Therefore, systems of the form in Eq.~(\ref{eq:ffhamiltonian}) may be considered \emph{exactly solvable} classically, since their exact diagonalization is reduced to exact diagonalization on a \emph{poly}$(n)$-sized matrix $\mathbf{h}$. 

\section{Fundamental Theorem}
\label{sec:fundthm}

As mentioned previously, we seek to characterize the full set of Jordan-Wigner-like transformations, generalizing Eq.~(\ref{eq:jwtransform}). 
To be more precise, we ask for the conditions under which there exists a mapping $\phi: V \mapsto \widetilde{V}^{\times 2}$, for some set $\widetilde{V}$ (the fermion modes), effecting
\begin{align}
    \sigma^{\boldsymbol{j}} \mapsto i \gamma_{\phi_1(\boldsymbol{j})} \gamma_{\phi_2(\boldsymbol{j})} \mathrm{,}
\label{eq:paulitopair}
\end{align}
for $\phi_1(\boldsymbol{j})$, $\phi_2(\boldsymbol{j}) \in \widetilde{V}$, and such that
\begin{align}
    \scomm{\sigma^{\boldsymbol{j}}}{\sigma^{\boldsymbol{k}}} = \scomm{\gamma_{\phi_1(\boldsymbol{j})} \gamma_{\phi_2(\boldsymbol{j})}}{\gamma_{\phi_1(\boldsymbol{k})} \gamma_{\phi_2(\boldsymbol{k})}}
\label{eq:scommeq}
\end{align}
for all pairs, $\boldsymbol{j}$ and $\boldsymbol{k}$. 
Such a mapping induces a term-by-term \emph{free-fermionization} of the Hamiltonian (\ref{eq:hdef}) to one of the form (\ref{eq:ffhamiltonian}) such that
\begin{align}
    G(H) \simeq G(\widetilde{H}) \mathrm{.}
\label{eq:grapheq}
\end{align}
Again, $G(H)$ is the frustration graph of $H$. 

From the canonical anticommutation relations, Eq.~(\ref{eq:car}), and the distribution rule Eq.~(\ref{eq:distribution}), we see that scalar commutators between quadratic Majorana-fermion operators are given by
\begin{align}
\scomm{\gamma_{\mu} \gamma_{\nu}}{\gamma_{\alpha} \gamma_{\beta}} = (-1)^{|(\mu, \nu) \cap (\alpha, \beta)|}
\label{eq:pairscomm}
\end{align}
Eqs.~(\ref{eq:grapheq}) and (\ref{eq:pairscomm}) can be restated graph theoretically as saying that $G(H)$ is the graph whose vertex set is the edge set of the fermion hopping graph $R$, and vertices of $G(H)$ are neighboring if and only if the associated edges of $R$ share exactly one vertex. 
Such a graph is called the \emph{line graph} of $R$.

\begin{definition}[Line Graphs]
The line graph $L(R) \equiv (E, F)$ of a root graph $R \equiv (V, E)$ is the graph whose vertex set is the edge set of $R$ and whose edge set is given by
\begin{align}
F \equiv \{(e_1, e_2) \ | \ e_1, e_2 \in E, \ |e_{1} \cap e_{2}| = 1\}
\end{align} 
That is, vertices are neighboring in $L(R)$ if the corresponding edges in $R$ are incident at a vertex.
\end{definition}

\noindent Notice that if $L(R)$ is connected if and only if $R$ is. 
With these definitions in-hand, our first main result can be stated simply as

\begin{theorem}[Existence of free-fermion solution] 
An injective map $\phi$ as defined in Eq.~(\ref{eq:paulitopair}) and Eq.~(\ref{eq:scommeq}) exists for the Hamiltonian $H$ as defined in Eq.~(\ref{eq:hdef}) if and only if there exists a root graph $R$ such that
\begin{align}
    G(H) \simeq L(R),
\label{eq:ffsolution}
\end{align} 
where R is the hopping graph of the free-fermion solution.
\label{thm:ffsolution}
\end{theorem}

\begin{table*}
\centering
\setcellgapes{3pt}
\makegapedcells
\begin{tabular} {c c c}
\toprule
& With Twins & Without twins \\ \midrule
\makecell{Forbidden\\ Graphs} & \makecell{(a)\vspace{16mm}} \makecell{\includegraphics[width=0.4\textwidth]{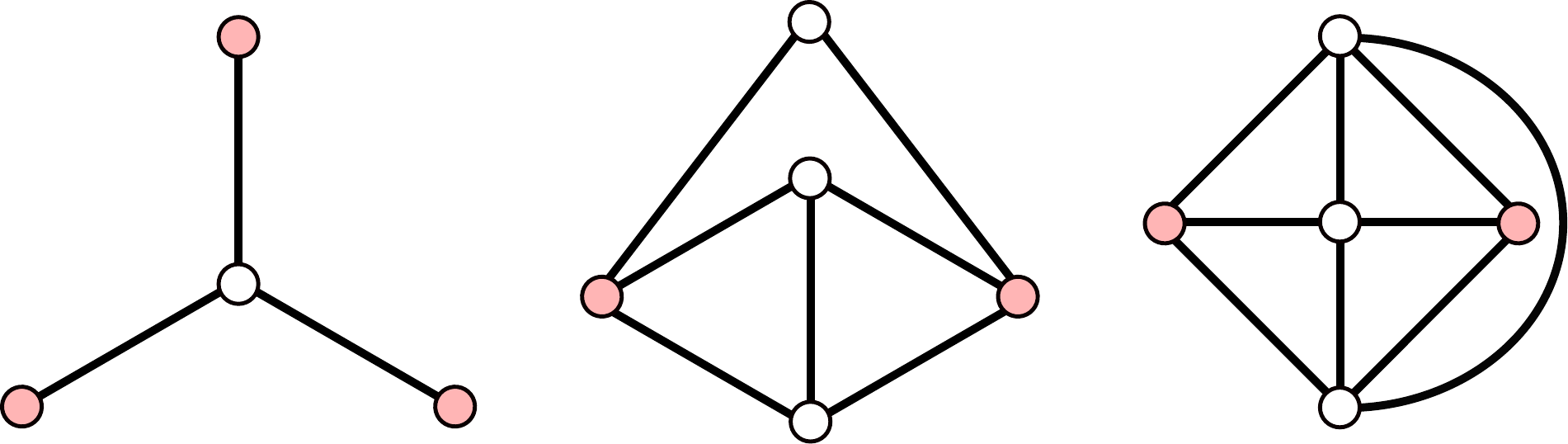}} & \makecell{(d)\vspace{16mm}} \hspace{5mm} \makecell{\multirow{2}{*}{\includegraphics[width=0.3\textwidth]{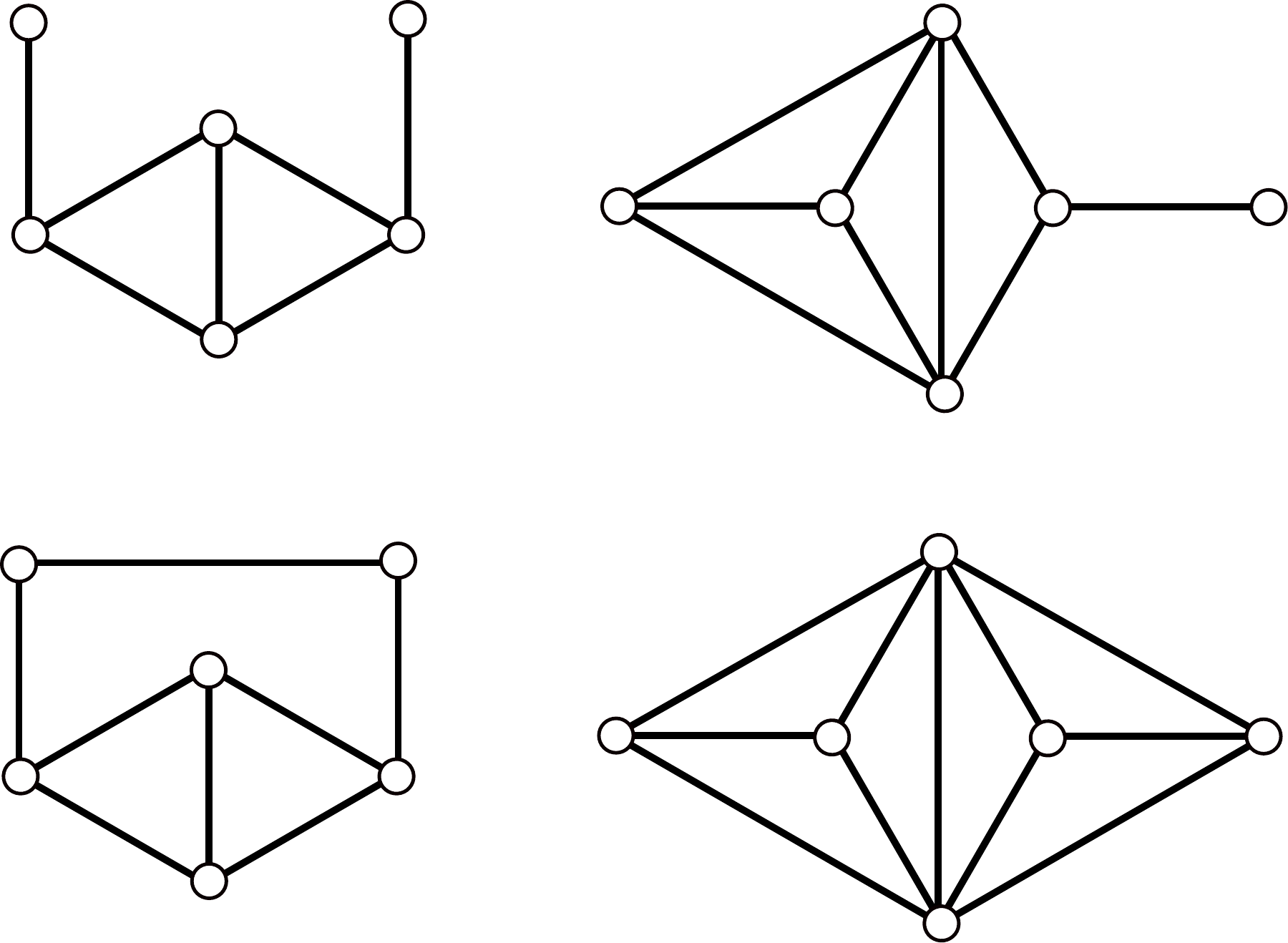}} \vspace{14mm}} \\ \cline{1-2}
\makecell{Twin-Free,\\ $L(R)$} &\makecell{(b)\vspace{10mm}} \hspace{8mm} \makecell{\includegraphics[width=0.33\textwidth]{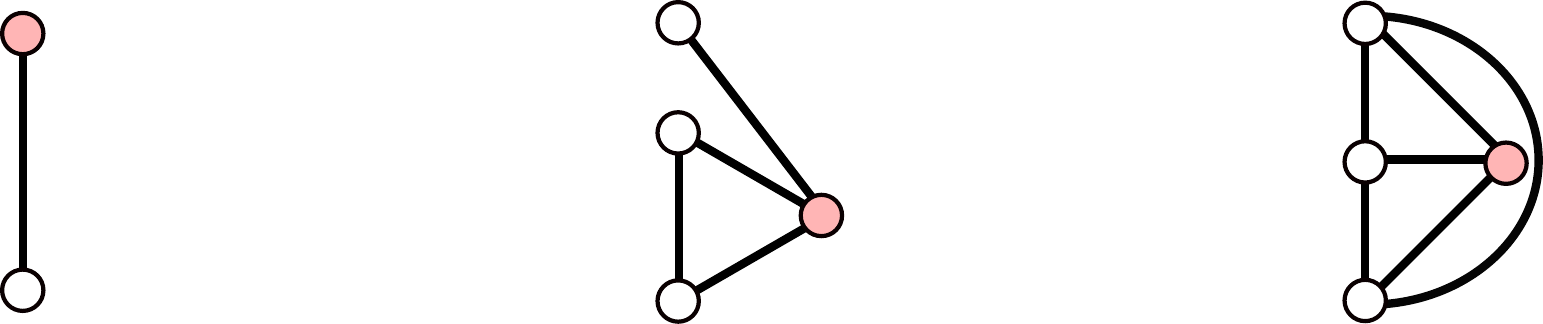}} \hspace{2mm} & \\ \cline{1-2}
\makecell{Root $R$} & \makecell{(c)\vspace{10mm}} \hspace{4mm} \makecell{\includegraphics[width=0.35\textwidth]{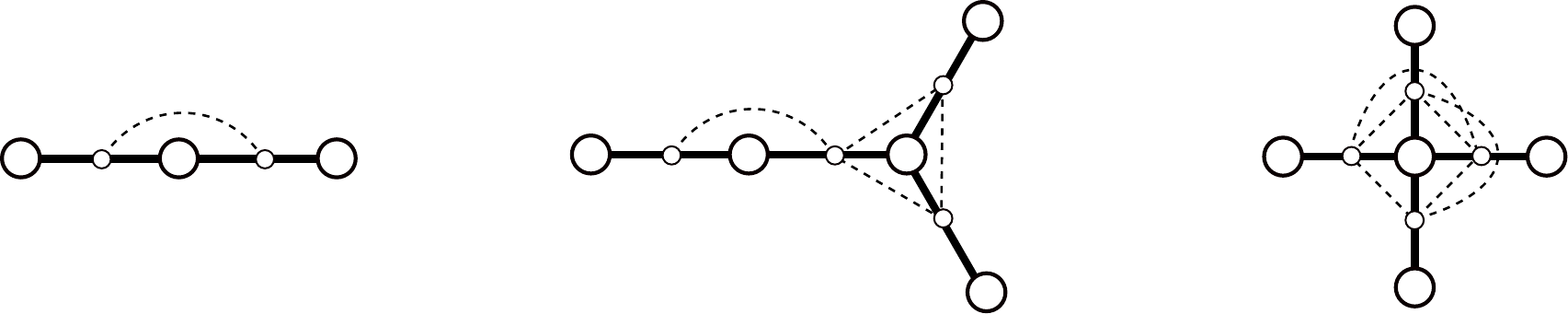}} \hspace{2mm} & \hspace{-5mm} \makecell{(e)\vspace{16mm}} \hspace{5mm}  \makecell{\vspace{-5mm}\includegraphics[width=0.29\textwidth]{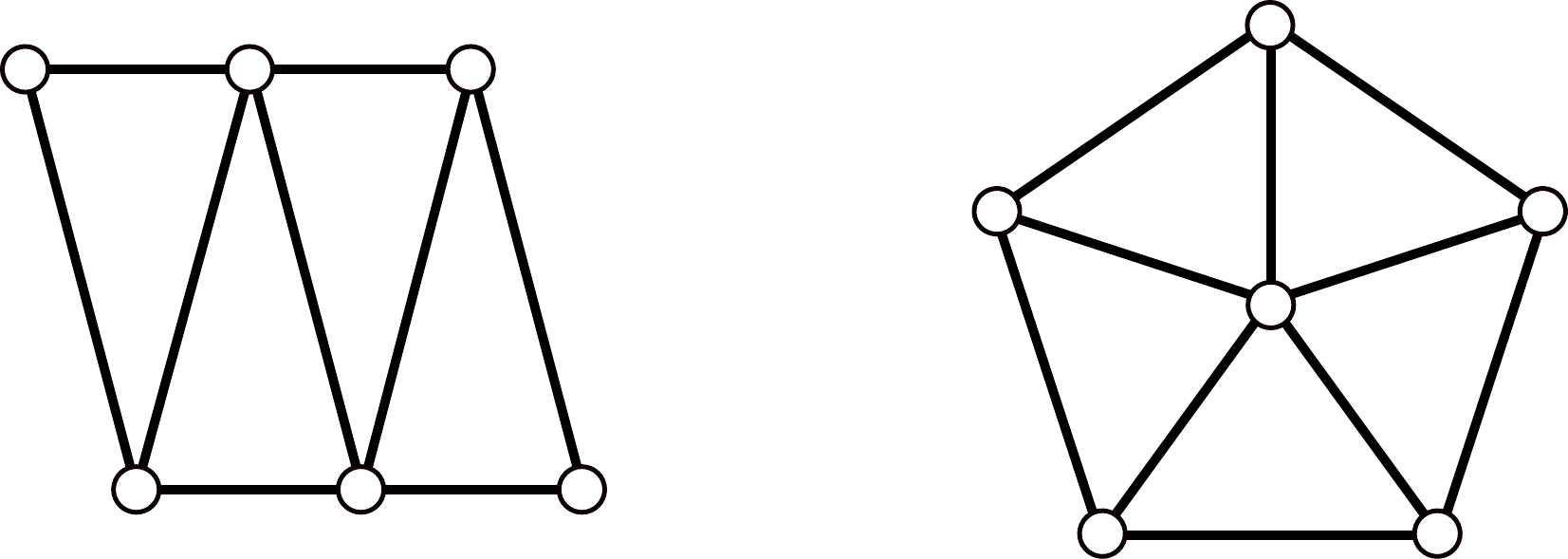}} \\ \bottomrule
\end{tabular}
\caption{A graph is a line graph if and only if it does not contain any of the nine forbidden graphs in (a), (d), and (e) as an induced subgraph \cite{beineke1970characterizations}. 
Of these nine graphs, the three in (a) contain twin vertices, highlighted. 
If these three graphs are induced subgraphs of a frustration graph such that these highlighted vertices are twins in the larger graph, then the twins can be removed by restricting onto a fixed mutual eigenspace of their products, which correspond to constants of motion of the Hamiltonian. 
(b) The twin-free restrictions of the graphs in (a), with all but one highlighted vertex from (a) removed. 
These graphs are the line graphs of the graphs in (c). 
In Ref.~\cite{soltes1994forbidden}, it was shown that only five graphs contain the forbidden subgraphs in (e) and none of those in (a) or (d). 
Finally, this set was further refined in Ref.~\cite{yang2002three} to a set of three forbidden subgraphs for 3-connected line graphs of minimum degree at least seven, though we do not display these graphs here.}
\label{tab:forbidden}
\end{table*}

\begin{proof}
The proof can be found in Section~\ref{sec:thm1proof}. 
\end{proof}

The intuition for this result is that the root graph $R$ is the graph where the vertices are fermions and the edges are the bilinears that appear in the Hamiltonian $H$. 
The result reveals a correspondence between a characterization of line graphs and a characterization of free-fermion spin models, as not every graph can be expressed as the line graph of some root. 
We must however note that, strictly speaking, the existence of this mapping alone does not guarantee a free-fermion solution, since the ``Lie-homomorphism" constraint, Eq.~(\ref{eq:scommeq}), does not fix the \emph{sign} of the terms in the free-fermion Hamiltonian. 
Choosing a sign for each term is equivalent to \emph{orienting} the root graph, since multiplying by a sign is equivalent to making the exchange $\phi_{1}(\boldsymbol{j}) \leftrightarrow \phi_{2}(\boldsymbol{j})$ in Eq.~(\ref{eq:paulitopair}). 
Different orientations may not faithfully reproduce the properties of $H$, but we will see that such an orientation can always be chosen. 
The line graph condition in Eq.~(\ref{eq:ffsolution}) is therefore necessary and sufficient for a free-fermion solution to exist. 
Before turning to further implications of Theorem~\ref{thm:ffsolution}, let us first detail some properties of line graphs.

Line graphs are closely related to so-called \emph{intersection graphs}, originally studied by Erd\H{o}s \cite{erdos1966representation} and others (see, for example, Ref.~\cite{harary1971graph}). 
An intersection graph $G \equiv (V, E)$ is a graph whose vertex set, $V \subseteq 2^{S}$, consists of distinct subsets of some set $S$. 
Two vertices, $u$ and $v$, are neighboring in $G$ if their intersection is nonempty ($|v \cap w| \neq 0$). 
A line graph is a special case of an intersection graph where every vertex corresponds to a subset of size at most two. 
When we specify that $\phi$ be injective, we are requiring that no distinct vertices have identical subsets, and our definition of a free-fermion solution Eq.~(\ref{eq:ffsolution}) identically coincides with that of a line graph. 
Since terms in $H$ can thus intersect by at most one Majorana mode, collections of terms containing a given mode are all neighboring in $G(H)$, so this mode corresponds to a \emph{clique}, or complete subgraph, of $G(H)$. 
This characterization of line graphs was first given by Krausz \cite{krausz1943demonstration} and bears stating formally.

\begin{definition}[Krausz decomposition of line graphs]
Given a line graph $G \simeq L(R)$, there exists a partition of the edges of $G$ into cliques such that every vertex appears in at most two cliques.
\end{definition} 

\noindent Cliques in $G(H)$ can therefore be identified with the individual Majorana modes in a free-fermion solution of $H$. 
If a term belongs to only one clique, we can  ensure our resulting fermion Hamiltonian is quadratic by taking the second clique for that term to be a clique of no edges, as we will see in several examples below. 
The existence of a Krausz decomposition is utilized in a linear-time algorithm to recognize line graphs by Roussopoulos \cite{roussopoulos1973max}, though the earliest such algorithm for line-graph recognition was given by Lehot \cite{lehot1974optimal}. 
A dynamic solution was later given by Degiorgi and Simon \cite{degiorgi1995dynamic}. 
These algorithms are optimal and constructive, and so can be applied to a given spin model to provide an exact free-fermion solution.

We next turn to the \emph{hereditary property} of line graphs, for which we require the following definition:

\begin{definition}[Induced subgraphs]
Given a graph $G \equiv (V, E)$, an induced subgraph of $G$ by a subset of vertices $V^{\prime} \subset V$, is a graph $G[V^{\prime}] \equiv (V^{\prime}, E^{\prime})$ such that for any pair of vertices $u$, $v \in V^{\prime}$, $(u, v) \in E^{\prime}$ if and only if $(u, v) \in E$ in $G$.
\end{definition} 

\noindent An induced subgraph of $G$ can be constructed by removing the subset of vertices $V/V^{\prime}$ from $G$, together with all edges incident to any vertex in this subset. 
Line graphs are a \emph{hereditary class} of graphs in the sense that any induced subgraph of a line graph is also a line graph. 
This coincides with our intuition that removing a term from a free-fermion Hamiltonian does not change its free-fermion solvability. 
Conversely, Hamiltonians for which no free-fermion solution exists are accompanied by ``pathological" structures in their frustration graphs, which obstruct a free-fermion description no matter how we try to impose one. 
This is captured by the forbidden subgraph characterization of Beineke \cite{beineke1970characterizations} and later refined by others \cite{soltes1994forbidden, yang2002three}.

\begin{corollary}[Beineke no-go theorem]
A given spin Hamiltonian $H$ has a free-fermion solution if and only if its frustration graph $G(H)$ does not contain any of nine forbidden subgraphs, shown in Table \ref{tab:forbidden}, (a) (d) (e), as an induced subgraph.
\label{cor:nogo}
\end{corollary}

\noindent These forbidden subgraphs above can be interpreted as collections of ``frustrating" terms. 
At least one of the terms must be assigned to a fermion interaction in every possible assignment from Pauli operators to fermions. 
Correspondingly, ignoring these terms by removing their corresponding vertices from the frustration graph may remove a forbidden subgraph and cause the Hamiltonian to become solvable. 
The terms which we need to remove in this way need not be unique. 
In the next section, we discuss one such strategy for removing vertices such that our solution will remain faithful to the original spin Hamiltonian by exploiting symmetries.

\section{Symmetries}
\label{sec:symmetries}

An important class of symmetries involves \textit{twin vertices} in the frustration graph.

\begin{definition}[Twin Vertices]
Given a graph $G \equiv (V, E)$, vertices $u$, $v \in V$ are twin vertices if, for every vertex $w \in V$, $(u, w) \in E$ if and only if $(v, w) \in E$. 
\end{definition}

\noindent Twin vertices have exactly the same neighborhood, and are thus never neighbors in a frustration graph, which contains no self edges due to the fact that every operator commutes with itself. 
Sets of twin vertices are the subject of our first lemma.

\begin{lemma}[Twin vertices are constants of motion]
Suppose a pair of terms $\sigma^{\boldsymbol{j}}$ and $\sigma^{\boldsymbol{k}}$ in $H$ correspond to twin vertices in $G(H)$, then the product $\sigma^{\boldsymbol{j}} \sigma^{\boldsymbol{k}}$ is a nontrivial Pauli operator commuting with every term in the Hamiltonian. 
Distinct such products therefore commute with each other.
\label{lem:twinvertices}
\end{lemma}

\begin{proof}
The statement follows straightforwardly from the definition of twin vertices: every term in $H$ (including $\sigma^{\boldsymbol{j}}$ and $\sigma^{\boldsymbol{k}}$ themselves) either commutes with both $\sigma^{\boldsymbol{j}}$ and $\sigma^{\boldsymbol{k}}$ or anticommutes with both of these operators. 
Terms in $H$ therefore always commute with the product $\sigma^{\boldsymbol{j}}\sigma^{\boldsymbol{k}}$. 
This product is furthermore a nontrivial Pauli operator, for if $\sigma^{\boldsymbol{j}} \sigma^{\boldsymbol{k}} = I$, then $\boldsymbol{j} = \boldsymbol{k}$, and we would not identify these Paulis with distinct vertices in $G(H)$. 
Constants of motion generated this way must commute with one another, since they commute with every term in the Hamiltonian and are themselves products of Hamiltonian terms. 
They therefore generate an abelian subgroup of the symmetry group of the Hamiltonian. 
\end{proof}

Let the symmetry subgroup generated by products of twin vertices in this way be denoted $\mathcal{S}$. 
We can leverage these symmetries to remove twin vertices from the frustration graph $G(H)$. 
To do this, choose a minimal generating set $\{\sigma^{\boldsymbol{s}}\}$ of Pauli operators for $\mathcal{S}$ and choose a $\pm 1$ eigenspace for each. 
Let $(-1)^{x_{\boldsymbol{s}}}$ be the eigenvalue associated to the generator $\sigma^{\boldsymbol{s}} \in \mathcal{S}$, for $x_{\boldsymbol{s}} \in \{0, 1\}$. 
We restrict to the subspace defined as the mutual $+1$ eigenspace of the stabilizer group
\begin{align}
    \mathcal{S}_{\boldsymbol{x}} = \langle (-1)^{x_{\boldsymbol{s}}} \sigma^{\boldsymbol{s}} \rangle
\end{align}
For a pair of twin vertices corresponding to Hamiltonian terms $\sigma^{\boldsymbol{j}}$ and $\sigma^{\boldsymbol{k}}$, we let
\begin{align}
    \sigma^{\boldsymbol{j}} \sigma^{\boldsymbol{k}} \equiv (-1)^{d_{\boldsymbol{j}, \boldsymbol{k}}} \left[\prod_{\boldsymbol{s} \in S_{\boldsymbol{j}, \boldsymbol{k}}} (-1)^{x_{\boldsymbol{s}}} \sigma^{\boldsymbol{s}} \right]
\end{align}
where $d_{\boldsymbol{j}, \boldsymbol{k}} \in \{0, 1\}$ specifies the appropriate sign factor, and $S_{\boldsymbol{j}, \boldsymbol{k}}$ is the subset of generators of $\mathcal{S}$ such that
\begin{align}
    \bigoplus_{\boldsymbol{s} \in S_{\boldsymbol{j}, \boldsymbol{k}}} \boldsymbol{s} = \boldsymbol{j} \oplus \boldsymbol{k}
\end{align}
where ``$\oplus$'' denotes addition modulo 2 here.
In the stabilizer subspace of $\mathcal{S}_{\boldsymbol{x}}$, we can make the substitution
\begin{align}
    \sigma^{\boldsymbol{k}} \rightarrow (-1)^{d_{\boldsymbol{j}, \boldsymbol{k}}} \sigma^{\boldsymbol{j}}
\end{align}
effectively removing the vertex $\boldsymbol{k}$ from $G(H)$.

Twin vertices capture the cases where a free-fermion solution for $H$ exists, but is necessarily non-injective. 
Indeed, note that we are careful in our statement of Theorem \ref{thm:ffsolution} to specify that our condition Eq.~(\ref{eq:ffsolution}) is necessary and sufficient when $\phi$ is injective. 
If we instead relax our requirement that vertices of a line graph correspond to distinct subsets of size two in our earlier discussion of intersection graphs, then we are allowing for line graphs of graphs with multiple edges, or \emph{multigraphs}. 
However, our definition of $G(\widetilde{H})$ will differ from the line graph of a multigraph for pairs of vertices corresponding to identical edges, which must be adjacent in the line graph of a multigraph, but will be nonadjacent in $G(\widetilde{H})$ from Eq.~(\ref{eq:pairscomm}). 
Such vertices will nevertheless be twin vertices in $G(\widetilde{H})$ due to the graph-isomorphism constraint, Eq.~(\ref{eq:grapheq}). 
Therefore, if no \emph{injective} mapping $\phi$ satisfying Theorem \ref{thm:ffsolution} exists, a many-to-one free-fermion solution exists only when twin vertices are present. 
Lemma~\ref{lem:twinvertices} allows us to deal with this non-injective case by removing twin vertices until we obtain the line graph of a simple graph when possible. The particular way we choose to perform this removal cannot affect the overall solvability of the model, since the frustration graph with all twin vertices removed is an induced subgraph of any frustration graph with only a proper subset of such vertices removed. A model which is solvable by free fermions this way is therefore solvable in all of its symmetry sectors.

Finally, we can see when a non-injective free-fermion solution may be possible from the forbidden subgraph characterization, Corollary~\ref{cor:nogo}. 
As seen in Table~\ref{tab:forbidden}, some of the forbidden subgraphs shown in (a) themselves contain twin vertices. 
If these forbidden subgraphs are connected to the global frustration graph such that their twin vertices remain twins in the larger graph, then they may be removed, possibly allowing for a solution of the full Hamiltonian by free fermions. 
When the twins are removed from the forbidden subgraphs, they become line graphs as shown in Table \ref{tab:forbidden} (b) and (c). An example of a model which can be solved this way is the Heisenberg-Ising model introduced in Ref.~\cite{lieb1961two}.

We next proceed to identify the remaining Pauli symmetries for a Hamiltonian satisfying Theorem \ref{thm:ffsolution}. 
For this, we invoke the natural partition of the Pauli group $\mathcal{P}$ into the subgroup $\mathcal{P}_{H}$, again defined as that generated by Hamiltonian terms $\{\sigma^{\boldsymbol{j}}\}_{\boldsymbol{j} \in V}$, and the Pauli operators outside this subgroup, $\mathcal{P}_{\perp} \equiv \mathcal{P}/\mathcal{P}_{H}$. 
Note that the latter set does not form a group in general, as for example, single-qubit Paulis may be outside of $\mathcal{P}_H$ yet may be multiplied to operators in $\mathcal{P}_H$. 
A subgroup of the symmetries of the Hamiltonian is the center $\mathcal{Z}(\mathcal{P}_H)$ of $\mathcal{P}_H$, the set of $n$-qubit Pauli operators in $\mathcal{P}_H$ which commute with every element of $\mathcal{P}_H$ and therefore with every term in the Hamiltonian. 
To characterize this group, we need two more definitions.

\begin{definition}[Cycle subgroup]
A cycle of a graph $G \equiv (V, E)$ is a subset of its edges, $Y \subseteq E$, such that every vertex contains an even number of incident edges from the subset. 
If a Pauli Hamiltonian satisfies Eq.~(\ref{eq:ffsolution}) for some root graph $R$, we define its cycle subgroup $Z_H \subseteq \mathcal{P}_H$ as the abelian Pauli subgroup generated by the cycles $\{Y_i\}_i$ of $R$,
\begin{align}
Z_H = \bigl\langle \Pi_{\{\boldsymbol{j}| \phi(\boldsymbol{j}) \in Y_i \}} \sigma^{\boldsymbol{j}} \bigr\rangle_i.
\label{eq:cyclesubgroup}
\end{align}   
\end{definition}

\noindent Since
\begin{align}
    \prod_{\{\boldsymbol{j}| \phi(\boldsymbol{j}) \in Y_i \}} \gamma_{\phi_1(\boldsymbol{j})} \gamma_{\phi_{2}(\boldsymbol{j})} = \pm I  
\end{align} 
we have, from Eq.~(\ref{eq:scommeq}) and the definition of $\phi$, that the elements of $Z_H$ commute with every term in the Hamiltonian and thus with each other (since they are products of Hamiltonian terms). 
That is, $Z_{H} \subseteq \mathcal{Z}(\mathcal{P}_H)$. 
Notice that the definition of the generators for $Z_H$ in Eq.~(\ref{eq:cyclesubgroup}) may sometimes yield operators proportional to identity.

A familiar symmetry of free-fermion Hamiltonians is the \emph{parity operator} 
\begin{align}
    P \equiv i^{\frac{1}{2}|\widetilde{V}| (|\widetilde{V}| - 1)}\prod_{k \in \widetilde{V}} \gamma_k 
\label{eq:fermparitydef}
\end{align}
which commutes with every term in the Hamiltonian since each term is quadratic in the Majorana modes. 
The phase factor is chosen such that $P$ is Hermitian. 
Here, we define this operator in terms of Pauli Hamiltonian terms through a combinatorial structure known as a T-join.

\begin{definition}[Parity operator]
A T-join of a graph $G \equiv (V, E)$ is a subset of edges, $T \subseteq E$, such that an odd number of edges from $T$ is incident to every vertex in $V$. 
If a Pauli Hamiltonian satisfies Eq.~(\ref{eq:ffsolution}) for some root graph $R$ such that the number of vertices in $R$ is even, we define the parity operator as
\begin{align}
    P \equiv  i^d \prod_{\boldsymbol{j} \in T} \sigma^{\boldsymbol{j}}
\label{eq:paritydef}
\end{align}
where the product is taken over a T-join of $R$, and $d \in \{0, 1, 2, 3\}$ specifies the phase necessary to agree with Eq.~(\ref{eq:fermparitydef}).
\end{definition}

\noindent Here we have
\begin{align}
    i^d \prod_{\boldsymbol{j} \in T} i \gamma_{\phi_1(\boldsymbol{j})}  \gamma_{\phi_2 (\boldsymbol{j})} &= i^{\frac{1}{2} |\widetilde{V}| (|\widetilde{V} | - 1)}\prod_{\mu \in \widetilde{V}} \gamma_{\mu} = P \mathrm{,}
\end{align} 
since every fermion mode will be hit an odd number of times in the T-join. 
Unlike with the cycle subgroup, $P$ is never proportional to the identity in the fermion description, though it may still be proportional to the identity in the Pauli description (up to stabilizer equivalences). 
In this case, only solutions for the free-fermion Hamiltonian in a fixed-parity subspace will be physical. We will see several examples of this in the next section.

When no T-join exists, we cannot form $P$ as a product of Hamiltonian terms. 
In fact, $P \in \mathcal{Z}(\mathcal{P}_H)$ only when $|\widetilde{V}|$ is even. 
Now with these definitions in hand, we are ready to state our second theorem.

\begin{theorem}[Symmetries are cycles and parity]
Given a Hamiltonian satisfying Eq.~(\ref{eq:ffsolution}) such that the number of vertices $|\widetilde{V}|$ in the root graph is odd, then we have
\begin{align}
    \mathcal{Z}(\mathcal{P}_H) = Z_{H}.
\end{align}
\noindent If the number of vertices in the root graph is even, then we have
\begin{align}
\mathcal{Z}(\mathcal{P}_H) = \left\langle Z_{H}, P \right\rangle.
\end{align}
\label{thm:symmetries}
\end{theorem}

\begin{proof}
The proof can be found in Section~\ref{sec:thm2proof}. 
\end{proof}

The Pauli symmetries of the Hamiltonian outside of $\mathcal{Z}(\mathcal{P}_H)$ may be thought of as ``logical" or ``gauge" qubits, and this characterization allows for a simple accounting of these qubits. 
Suppose we express a spin Hamiltonian $H$ on $n$ qubits as a free-fermion Hamiltonian on the hopping graph $R = (\widetilde{V}, \widetilde{E})$, and let $|\mathcal{Z}(\mathcal{P}_H)|$ be the number of independent generators of $\mathcal{Z}(\mathcal{P}_H)$. 
The number of logical qubits $n_L$ of the model is given by
\begin{align}
    n_L \equiv 
    \begin{cases}
        n - \left[\frac{1}{2} (|\widetilde{V}| - 1) + |\mathcal{Z}(\mathcal{P}_H)|\right]  & |\widetilde{V}| \ \mathrm{odd} \\
        n - \left[\frac{1}{2} (|\widetilde{V}| - 2) + |\mathcal{Z}(\mathcal{P}_H)|\right] & |\widetilde{V}| \ \mathrm{even}
    \end{cases} .
\end{align} 
This follows from the fact that the $\mathds{F}_2$-rank of the adjacency matrix of $G(H)$ is twice the number of qubits spanned by the fermionic degrees of freedom in the model, and also the number of vertices of the root graph $R$ up to a constant shift.

\begin{table}
\centering
\setcellgapes{3pt}
\makegapedcells
\begin{tabular}{ c c }
\toprule
$R$ & $L(R)$ \\ 
\midrule
\makecell{\includegraphics[width=0.19\textwidth]{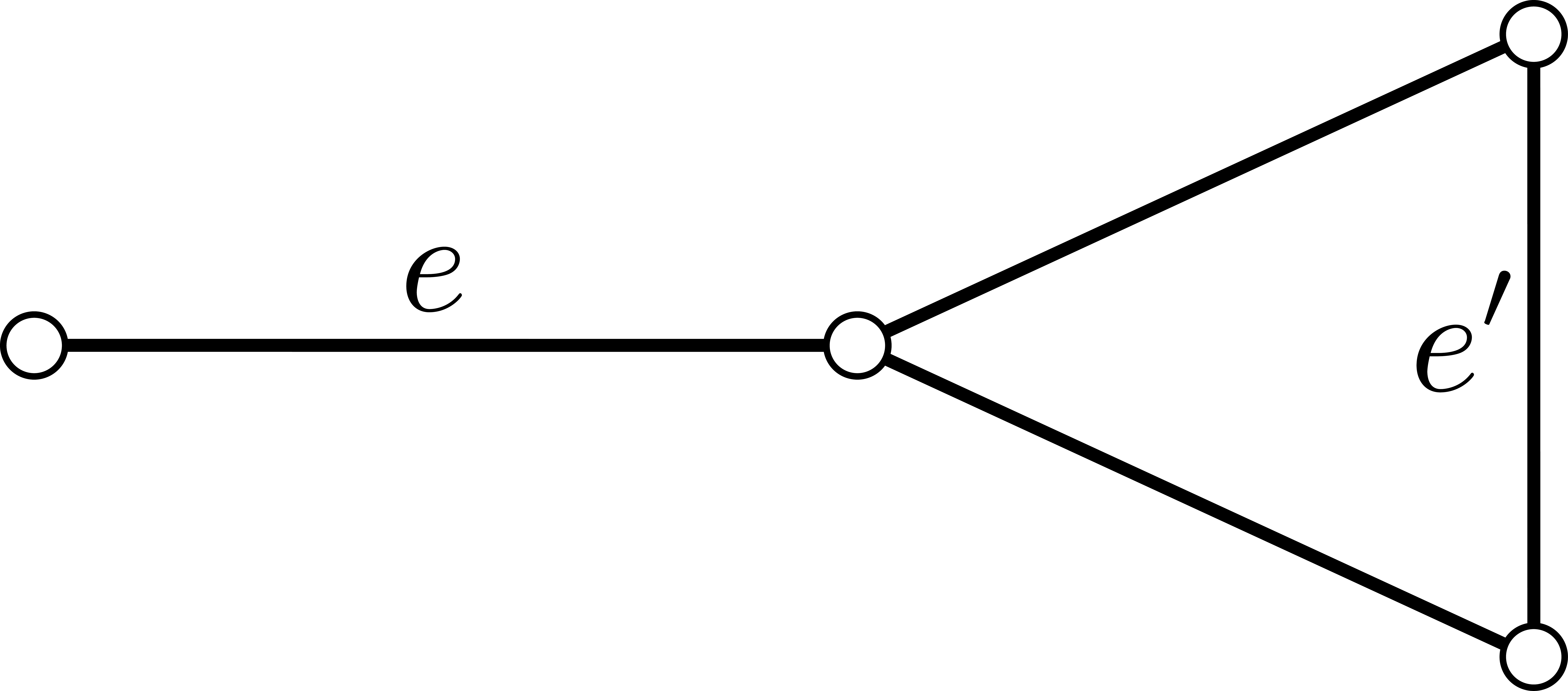}} &  \makecell{\includegraphics[width=0.18\textwidth]{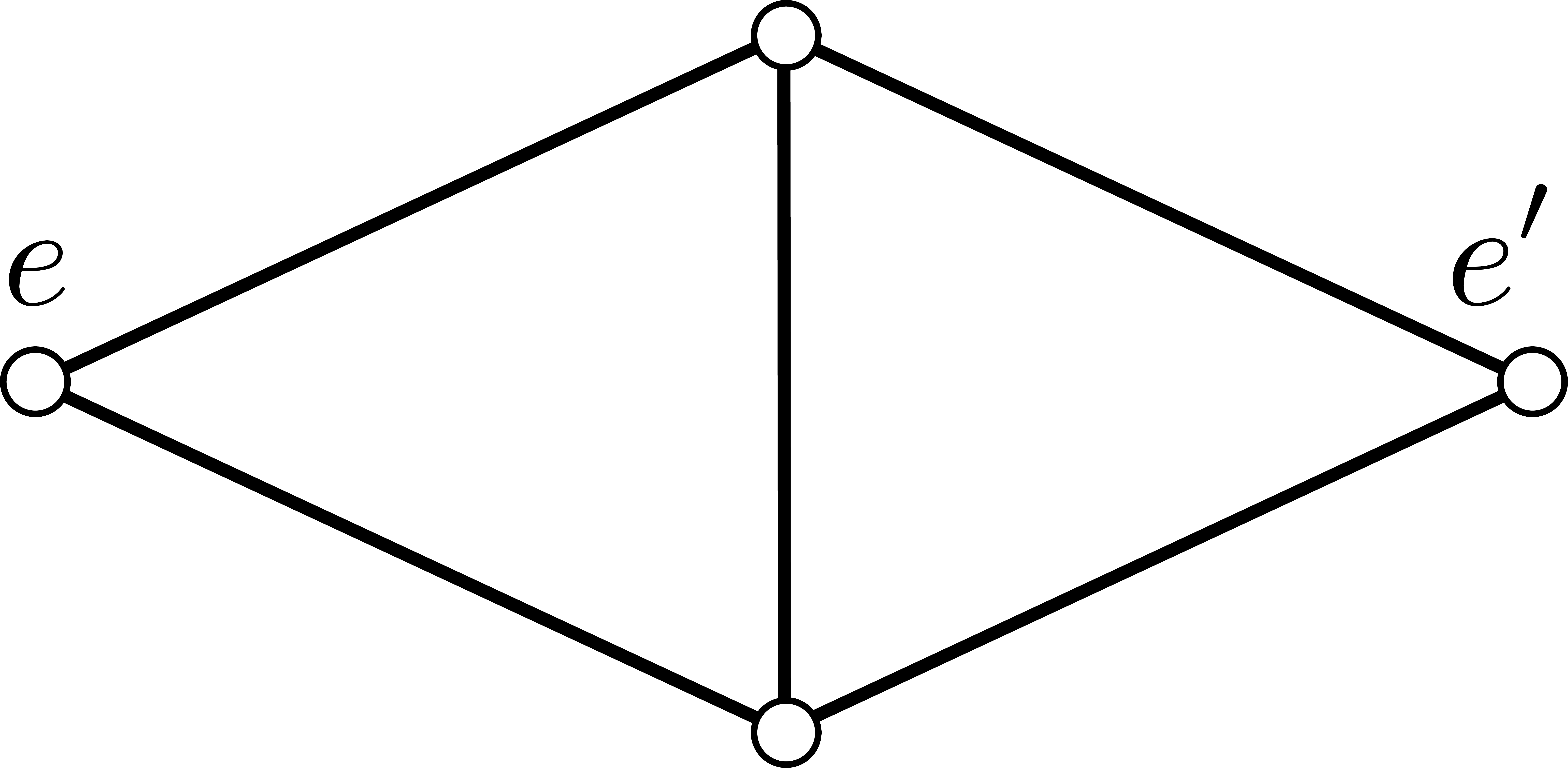}} \\ 
\makecell{\includegraphics[width=0.19\textwidth]{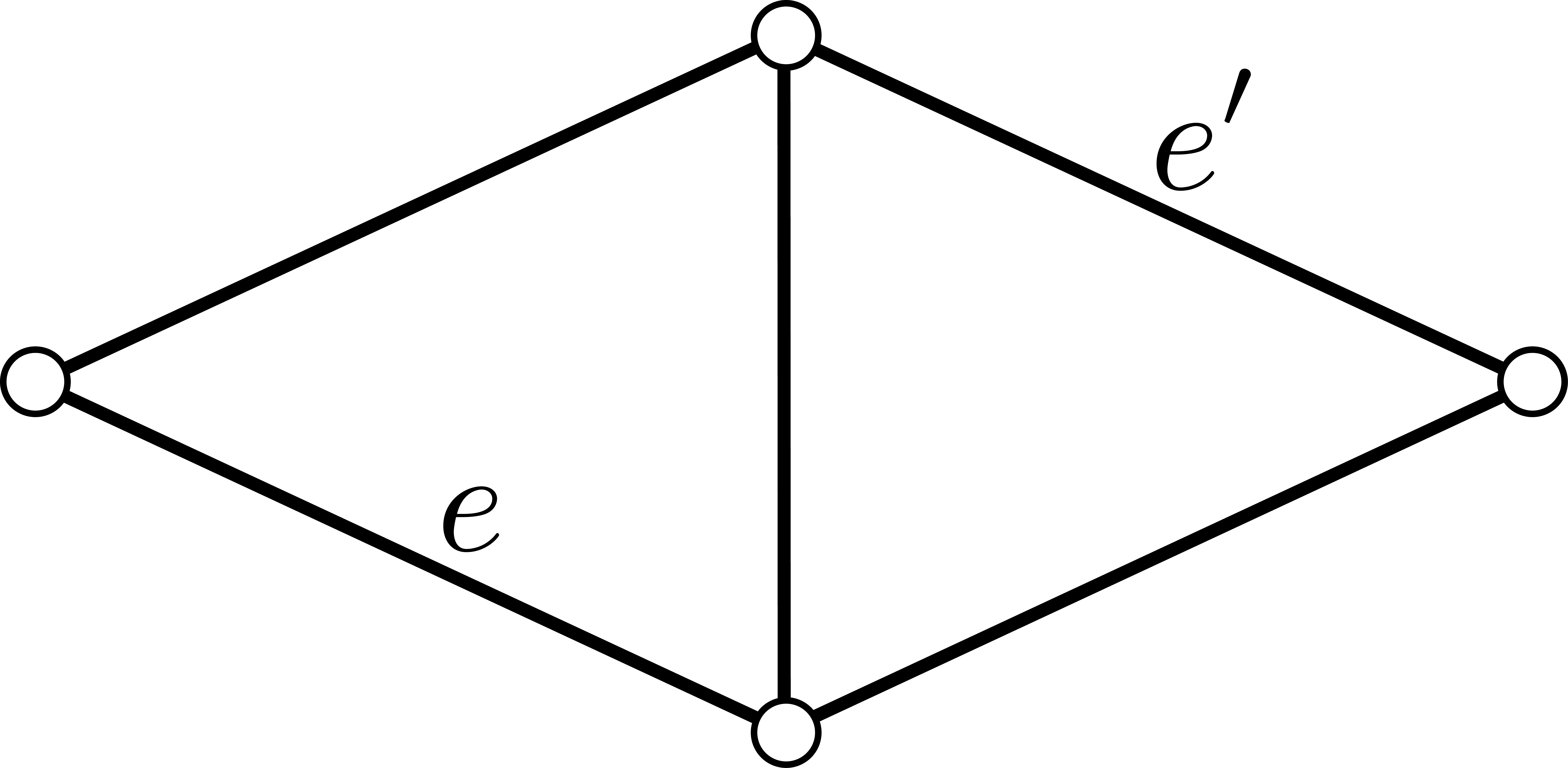}} &  \makecell{\includegraphics[width=0.18\textwidth]{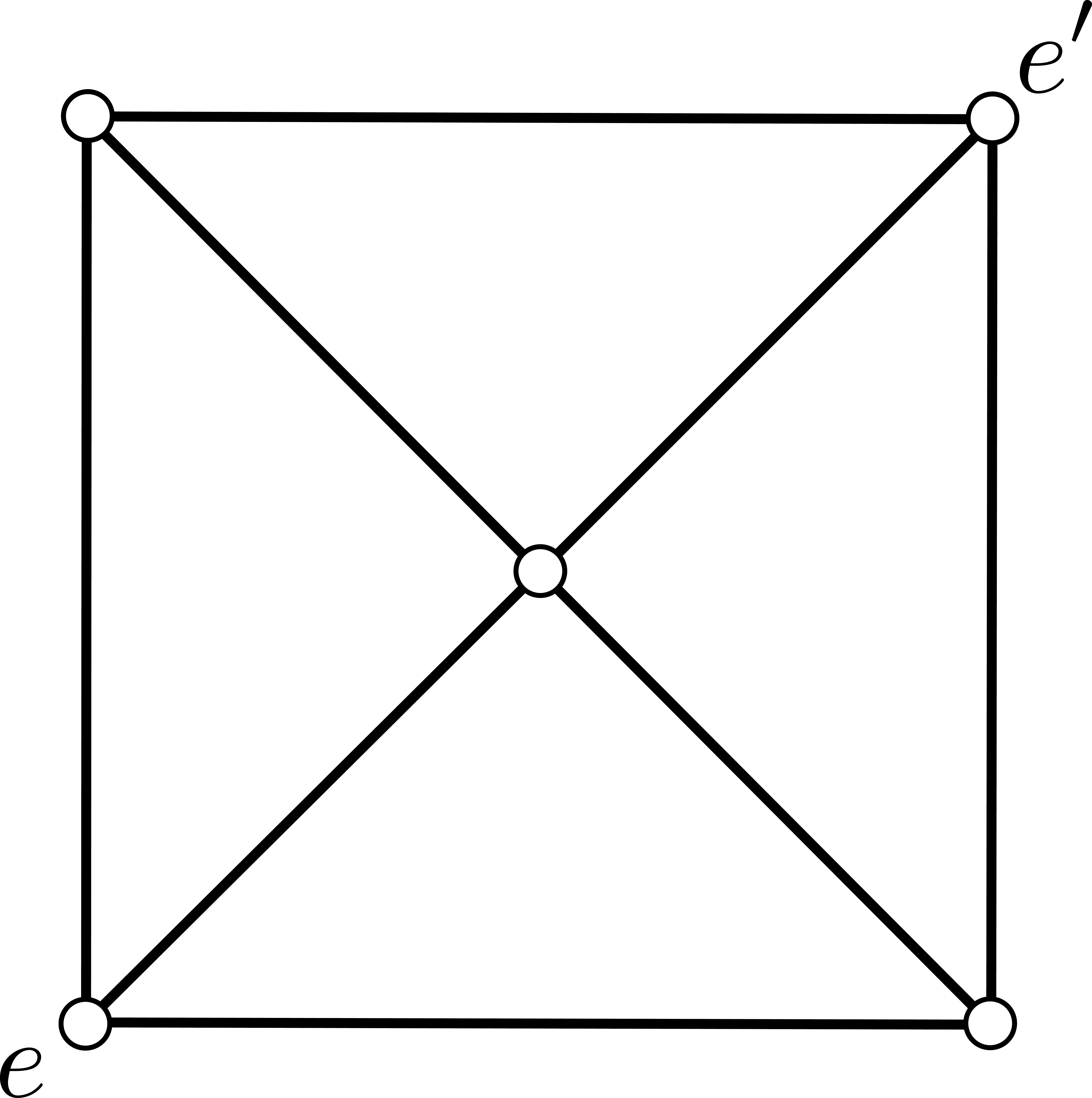}} \\ 
\makecell{\includegraphics[width=0.19\textwidth]{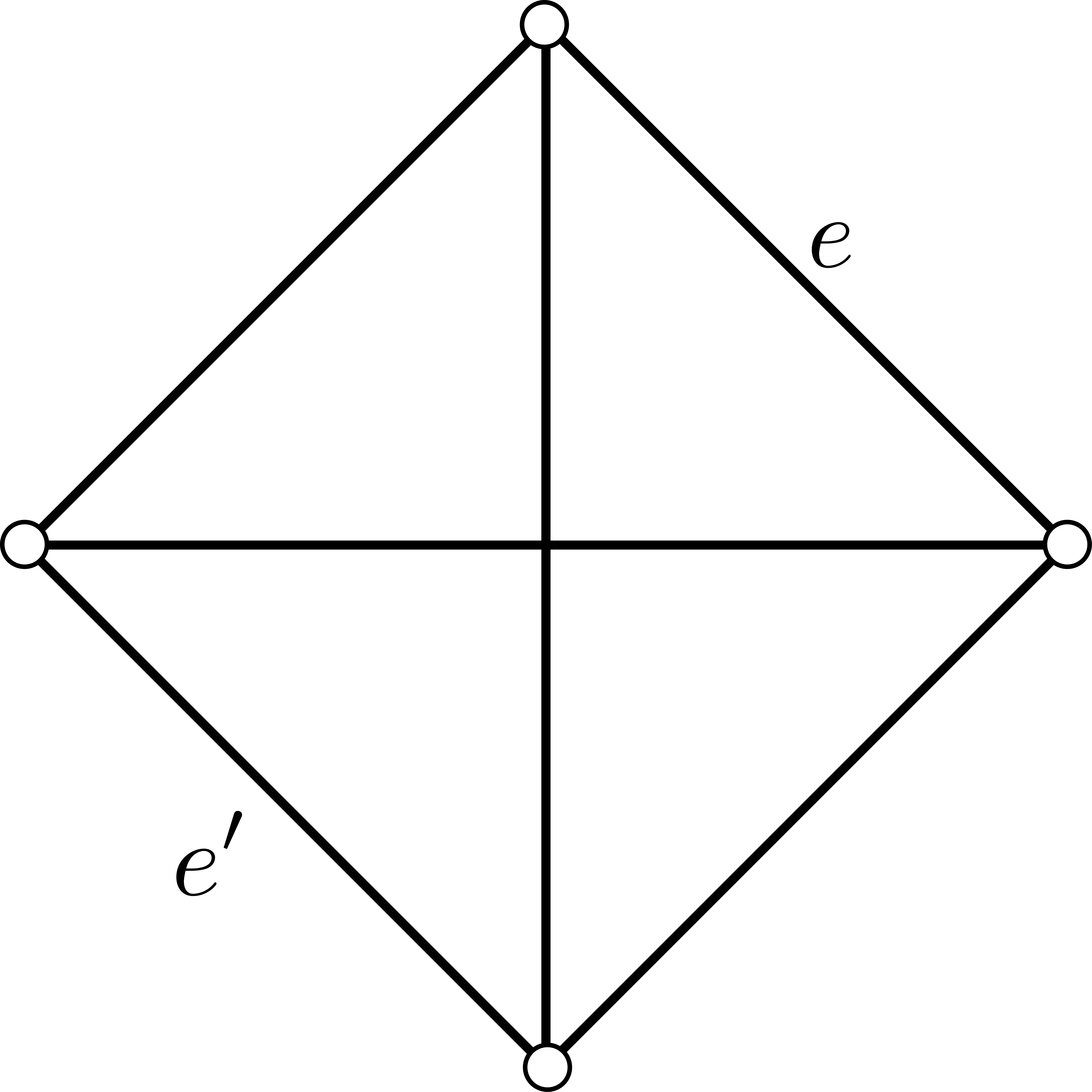}} &  \makecell{\includegraphics[width=0.18\textwidth]{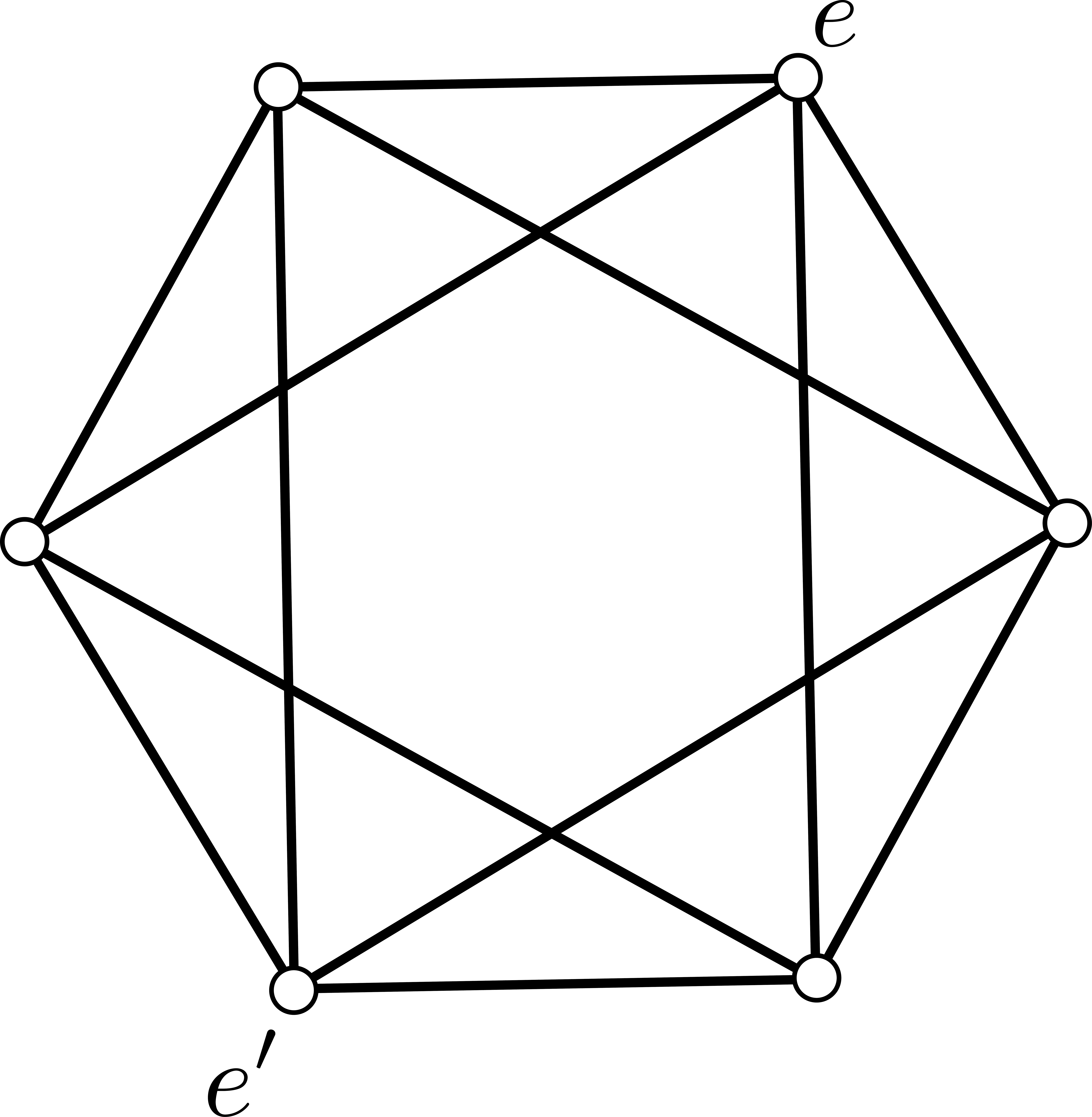}} \\ 
\bottomrule
\end{tabular}
\caption{The Whitney isomorphism theorem \cite{whitney1932congruent} guarantees that the edge automorphisms exchanging $e$ and $e^{\prime}$ in the graphs $R$ in the left column, or corresponding vertices in their line graphs on the right, which cannot be realized by any vertex automorphism of $R$, are the only such cases.}
\label{tab:strongexcept}
\end{table}

Finally, we note that there may be additional symmetries, such as translation invariance, if the coefficients $h_{\boldsymbol{j}}$ themselves satisfy a symmetry. 
Our characterization will allow us to say something about this situation when the associated symmetry transformation is a Clifford operator---that is, a unitary operator in the normalizer of the Pauli group---commuting with the Pauli symmetries in $\mathcal{Z}(\mathcal{P}_H)$, such as, e.g., a spatial translation. 
The following statement follows from a theorem by Whitney \cite{whitney1932congruent} (and extended to infinite graphs in \cite{bednarek1985whitneys}).

\begin{customcor}{1.2}[Clifford Symmetries and Whitney Isomorphism]
Let $\widetilde{H} = i \boldsymbol{\gamma} \cdot \mathbf{h} \cdot \boldsymbol{\gamma}^{\mathrm{T}}$ be a free-fermion Hamiltonian with single-particle Hamiltonian $\mathbf{h}$ in a fixed symmetry sector of $\mathcal{Z}(\mathcal{P}_H)$. 
Then any unitary Clifford symmetry $U$ such that $U^\dag \widetilde{H}U = \widetilde{H}$ induces a signed permutation symmetry $\mathbf{u}$ such that $\mathbf{h}=\mathbf{u}^{\mathrm{T}}\cdot\mathbf{h}\cdot\mathbf{u}$, except for when $U$ induces one of the three edge isomorphisms shown in Table~\ref{tab:strongexcept}.
\label{cor:clifsym}
\end{customcor}

\begin{proof}
This follows from the Whitney isomorphism theorem: except for the three cases shown in Table~\ref{tab:strongexcept}, any adjacency-preserving permutation of the vertices of $G(\widetilde{H})$ is induced by an adjacency-preserving permutation of the vertices of $R$. 
A Clifford symmetry $U$ acts as a signed permutation of the Hamiltonian terms which preserves $\widetilde{H}$. 
Suppose the associated unsigned permutation is not one of the exceptional cases, and so is induced by a permutation $\pi$ on the vertices of $R$. 
This gives 
\begin{align}
    \widetilde{H} &= U^{\dagger} \widetilde{H} U \label{eq:cor1p2proofline1} \\
    &= i \sum_{(j, k) \in \widetilde{E}} h_{jk} \left(U^{\dagger} \gamma_j U \right) \left(U^{\dagger} \gamma_k U\right) \label{eq:cor1p2proofline2}  \\
    &= i \sum_{(j, k) \in \widetilde{E}} (-1)^{x_j + x_k} h_{jk} \gamma_{\pi(j)} \gamma_{\pi(k)} \label{eq:cor1p2proofline3}
\end{align}
where $x_{j} \in \{0, 1\}$ designates the sign associated to the permutation of vertex $j \in \widetilde{V}$. 
By unitarity, this sign must depend on $j$ alone, since $U^{\dagger} \gamma_{j} U$ can only depend on $j$. 
Let $\mathbf{u}$ be a single-particle transition matrix defined as
\begin{align}
    u_{jk} = (-1)^{x_j} \delta_{k \pi(j)}.
\end{align}
Then we can reinterpret Eq.~(\ref{eq:cor1p2proofline3}) in the single-particle picture as
\begin{align}
    i \boldsymbol{\gamma} \cdot \mathbf{h} \cdot \boldsymbol{\gamma}^{\mathrm{T}} &= i \boldsymbol{\gamma} \cdot \left(\mathbf{u}^{\mathrm{T}} \cdot \mathbf{h} \cdot \mathbf{u} \right) \cdot \boldsymbol{\gamma}^{\mathrm{T}}. \label{eq:cor1p2proofline4}
\end{align}
By linear independence, Eq.~(\ref{eq:cor1p2proofline4}) therefore implies
\begin{align}
    \mathbf{h} = \mathbf{u}^{\mathrm{T}} \cdot \mathbf{h} \cdot \mathbf{u}
\end{align}
and the claim follows.
\end{proof}

See section \ref{sec:small} for a simple example of an exceptional Hamiltonian realizing a frustration graph shown in Table \ref{tab:strongexcept}. We now complete our characterization of free-fermion solutions by choosing an orientation for every edge in the root graph over a restricted subspace determined by the constants of motion.

\subsection{Orientation and Full Solution}
\label{sec:orientation}

As discussed previously, the Lie-homomorphism condition Eq.~(\ref{eq:scommeq}) does not fully constrain the free-fermion solution of a given Pauli Hamiltonian. 
This is because we are free to choose a direction to each edge in the root graph by exchanging $\phi_1(\boldsymbol{j}) \leftrightarrow \phi_2(\boldsymbol{j})$, which is equivalent to changing the sign of the term $i \gamma_{\phi_1(\boldsymbol{j})} \gamma_{\phi_2(\boldsymbol{j})}$ in $\widetilde{H}$ corresponding to $\sigma^{\boldsymbol{j}}$ in $H$. 
A related ambiguity corresponds to the cycle symmetry subgroup $Z_H$: we are free to choose a symmetry sector over which to solve the Hamiltonian $H$ by choosing a mutual $\pm 1$-eigenspace of independent nontrivial generators of this group. 
It will turn out that both ambiguities are resolved simultaneously.

First suppose we have a Hamiltonian $H$ satisfying Eq.~(\ref{eq:ffsolution}) for some root graph $R \equiv (\widetilde{V}, \widetilde{E})$. 
Construct a spanning tree $\Upsilon \equiv (\widetilde{V}, \widetilde{E}^{\prime})$ of $R$, defined as:

\begin{definition}[Spanning Tree]
Given a connected graph $G \equiv (V, E)$, a spanning tree $\Upsilon \equiv (V, E^{\prime}) \subseteq G$ is a connected subgraph of $G$ such that $E^{\prime}$ contains no cycles. 
\end{definition}

\noindent This can be performed in linear time in $|\widetilde{V}|$. 
Designate a particular vertex $v \in \widetilde{V}$ as the root of this tree. 
Each vertex $u \in \widetilde{V}$ has a unique path $p(u, v) \subseteq \widetilde{E}$ in $\Upsilon$ to the root, the path $p(v, v)$ being empty. 
Choose an arbitrary direction for each edge in $\widetilde{E}^{\prime}$ (we will see shortly to what extent this choice is important). 

Our choice of spanning tree determines a basis of \emph{fundamental cycles} for the binary cycle space of $R$ and thus a generating set of Paulis for the cycle subgroup $Z_H$. 
To see this, note that for each edge $\phi(\boldsymbol{j}) \in \widetilde{E}/\widetilde{E}^{\prime}$, there is a unique cycle of $R$ given by 
\begin{align}
    Y_{\boldsymbol{j}} \equiv p[\phi_1(\boldsymbol{j}), v] \cup p[\phi_2(\boldsymbol{j}), v] \cup \phi(\boldsymbol{j}).
\end{align}
Let $\sigma^{\boldsymbol{y}(\boldsymbol{j})}$ be the cycle subgroup generator associated to $Y_{\boldsymbol{j}}$, defined by
\begin{align}
    \boldsymbol{y}(\boldsymbol{j}) = \bigoplus_{\{\boldsymbol{z} | \phi(\boldsymbol{z}) \in Y_{\boldsymbol{j}}\}} \boldsymbol{z} 
\label{eq:cyclegendef}
\end{align}
such that
\begin{align}
    \sigma^{\boldsymbol{y}(\boldsymbol{j})} = i^d \prod_{\{\boldsymbol{z} | \phi(\boldsymbol{z}) \in Y_{\boldsymbol{j}}\}} \sigma^{\boldsymbol{z}}
\label{eq:cyclesubgroupgendef}
\end{align}
where $d \in \{0, 1, 2, 3\}$ again designates the appropriate phase. 
The set of such $Y_{\boldsymbol{j}}$ contains $|\widetilde{E}| - |\widetilde{V}| + 1$ cycles and forms an independent generating set for all the cycles of $R$ under symmetric difference. 
The corresponding set of $\sigma^{\boldsymbol{y}(\boldsymbol{j})}$ is therefore an independent generating set of the cycle subgroup up to signs, since individual Pauli operators either commute or anticommute and square to the identity. 

In a similar fashion as with twin-vertex symmetries, we restrict to a mutual $\pm 1$ eigenspace of the cycle-subgroup generators, designated by a binary string $\boldsymbol{x} \in \{0, 1\}^{\times |\widetilde{E}| - |\widetilde{V}| + 1}$ over the $\boldsymbol{j}$ such that $\phi(\boldsymbol{j}) \in \widetilde{E}/\widetilde{E}^{\prime}$. 
That is, we restrict to the mutual $+1$ eigenspace of the stabilizer group
\begin{align}
    Z_{H, \boldsymbol{x}} \equiv \bigl\langle (-1)^{x_{\boldsymbol{j}}} \sigma^{\boldsymbol{y}(\boldsymbol{j})}\bigr\rangle.
\end{align}
If Eq.~(\ref{eq:cyclegendef}) gives $\boldsymbol{y}(\boldsymbol{j}) = \mathbf{0}$ for any $\boldsymbol{j}$, then we take the corresponding $x_{\boldsymbol{j}} = 0$. 
We then simply choose the direction for the edge $\phi(\boldsymbol{j})$ such that
\begin{align}
    (-1)^{x_{\boldsymbol{j}}} i^d \left[\prod_{\{\boldsymbol{z} | \phi(\boldsymbol{z}) \in Y_{\boldsymbol{j}}\}} i \gamma_{\phi_1(\boldsymbol{z})} \gamma_{\phi_2(\boldsymbol{z})}\right] = +I
\label{eq:directioncond}
\end{align}
where $d$ is as defined in Eq.~(\ref{eq:cyclesubgroupgendef}). 
This ensures that the product of Majorana hopping terms around a fundamental cycle $Y_{\boldsymbol{j}}$ agrees with the corresponding Pauli product over the restricted subspace (i.e.\ up to equivalencies by stabilizers in the group $Z_{H, \boldsymbol{x}}$). 
By the Lie-homomorphism constraint Eq.~(\ref{eq:scommeq}), all products of Majorana hopping terms around a cycle of $R$ therefore agree with their corresponding Pauli products over this subspace, and so the multiplication relations of the Paulis are respected by their associated fermion hopping terms up to stabilizer equivalencies. 
Since we have exactly as many elements $Y_{\boldsymbol{j}}$ in our fundamental cycle basis as undirected edges $\phi(\boldsymbol{j})$, such an orientation can always be chosen.

\begin{table*}[t]
\centering
\setcellgapes{20pt}
\makegapedcells
\begin{tabular} {l c c c c c}
\toprule
\makecell{$\sigma^{i}_1$\textbackslash$\sigma^{j}_2$} & \makecell{$I$} & \makecell{$X$} & \makecell{$Y$} & \makecell{$Z$} & 2-qubit frustration graph $L(K_6)$ \\ \midrule
\makecell{$I$} & \makecell{$P \equiv i \gamma_1 \gamma_2 \gamma_3 \gamma_4 \gamma_5 \gamma_6$} & \makecell{$i \gamma_3 \gamma_5$} & \makecell{$i \gamma_2 \gamma_5$} & \makecell{$i \gamma_2 \gamma_3$} & \makecell{\multirow{4}{*}{\includegraphics[width=0.4\textwidth]{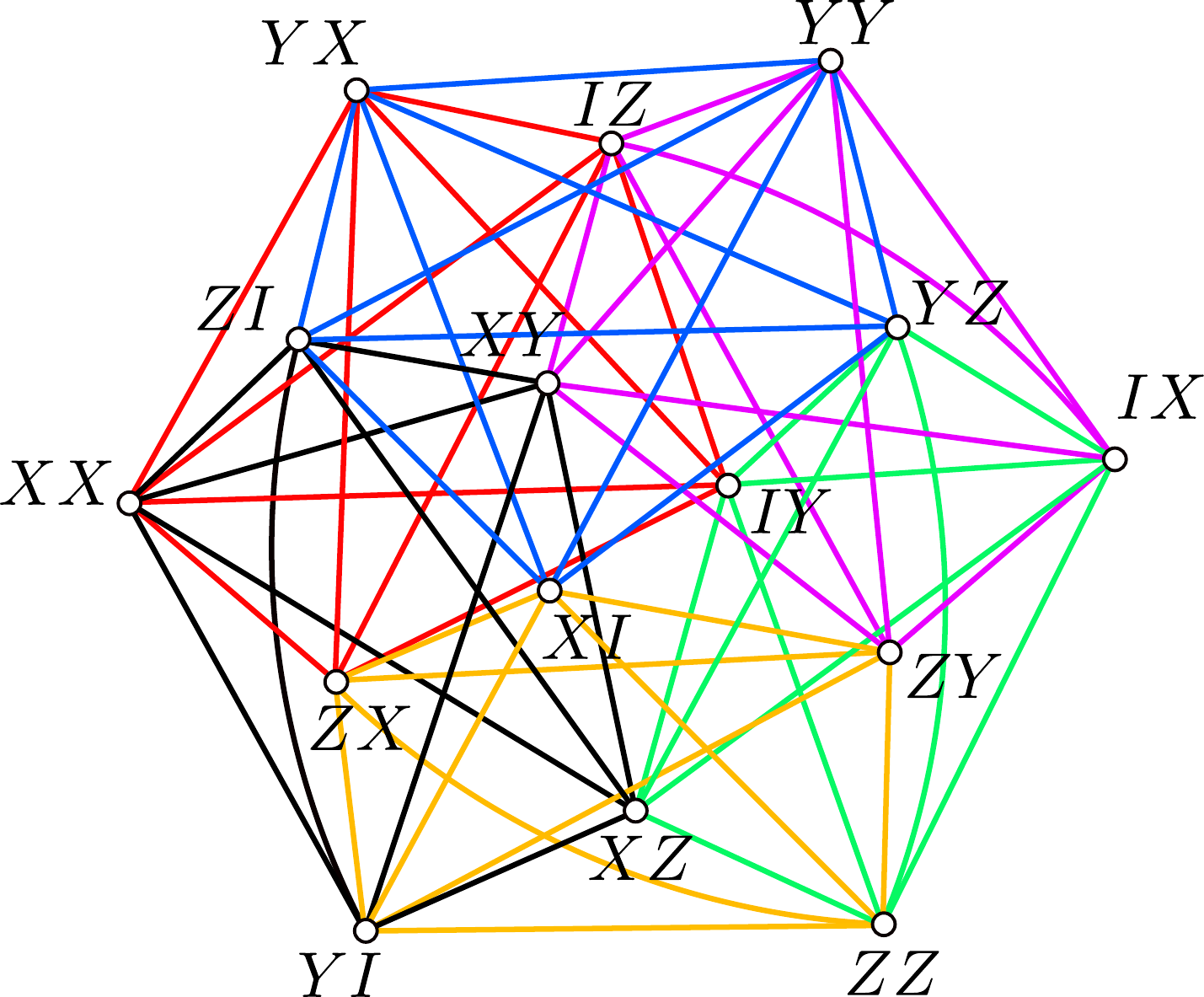}}} \\ \cline{1-5}
\makecell{$X$} & \makecell{$i \gamma_4 \gamma_6$} & \makecell{$i \gamma_1 \gamma_2$} & \makecell{$-i \gamma_1 \gamma_3$} & \makecell{$i \gamma_1 \gamma_5$} & \\ \cline{1-5}
\makecell{$Y$} & \makecell{$-i \gamma_1 \gamma_6$} & \makecell{$-i \gamma_2 \gamma_4$} & \makecell{$i \gamma_3 \gamma_4$} & \makecell{$i \gamma_4 \gamma_5$} & \\
\cline{1-5}
\makecell{$Z$} & \makecell{$-i \gamma_1 \gamma_4$} & \makecell{$i \gamma_2 \gamma_6$} & \makecell{$-i \gamma_3 \gamma_6$} & \makecell{$i \gamma_5 \gamma_6$} & \\ \bottomrule
\end{tabular}
\caption{(Left) Fermionization of the two-qubit Pauli algebra $\mathcal{P}_2 = \{\sigma^i_1 \otimes \sigma^j_2 \}_{(i, j) \neq (0, 0)}$ by six fermion modes. 
The graph isomorphism $G(\mathcal{P}_2) \simeq L(K_6)$ reflects the Lie-algebra isomorphism between $\mathfrak{su}(4)$ and $\mathfrak{spin}(6)$. 
Though the scalar-commutation relations are reproduced by quadratics in $\{\gamma_{\mu}\}_{\mu = 1}^6$, the one-sided multiplication relations are only recovered upon projecting onto the $+1$ eigenspace of $P \equiv i \gamma_1 \gamma_2 \gamma_3 \gamma_4 \gamma_5 \gamma_6$ on the fermion side of the mapping. 
(Right) The graph $L(K_6)$, with vertices labeled by a particular satisfying Pauli assignment. 
Edges are colored to identify the six $K_5$ subgraphs in the Krausz decomposition of this graph, corresponding to the six fermion modes. 
Each vertex belongs to exactly two such subgraphs, as must be the case for a line graph. 
This graphical correspondence was first observed in Ref.~\cite{goodmanson1996graphical} in the language of the Dirac algebra.}
\label{tab:2qubit}
\end{table*}

Why are we free then, to choose an arbitrary sign for each free-fermion term in our original spanning tree? 
This choice is actually equivalent to a choice of signs on the definitions of the individual Majorana modes themselves and so amounts to a choice of orientation for the coordinate basis in which we write $\mathbf{h}$. 
To see this, choose a fiducial orientation for $R$ satisfying Eq.~(\ref{eq:directioncond}), and suppose our particular free-fermion solution---not necessarily oriented this way---corresponds to the mapping
\begin{align}
    \sigma^{\boldsymbol{j}} \mapsto i (-1)^{x_{\boldsymbol{j}}} \gamma_{\phi_1(\boldsymbol{j})} \gamma_{\phi_2(\boldsymbol{j})}
\end{align}
for $\phi(\boldsymbol{j}) \in \widetilde{E}^{\prime}$ and with $x_{\boldsymbol{j}} \in \{0, 1\}$ designating the edge-direction of $\phi(\boldsymbol{j})$ relative to the fiducial orientation. 
We have
\begin{align}
    (-1)^{x_{\boldsymbol{j}}} &= (-1)^{r_{\phi_1(\boldsymbol{j})} + r_{\phi_2(\boldsymbol{j})}}
\label{eq:phasekickback}
\end{align}
where
\begin{align}
    r_{u} = \sum_{\{\boldsymbol{k}| \phi(\boldsymbol{k}) \in p(u, v)\}} x_{\boldsymbol{k}}.
\end{align}
Since the symmetric difference of $p[\phi_1(\boldsymbol{j}), v]$ and $p[\phi_2(\boldsymbol{j}), v]$ is the edge $\phi(\boldsymbol{j}) \in \widetilde{E}^{\prime}$, all sign factors on the right side of Eq.~(\ref{eq:phasekickback}) cancel except for $(-1)^{x_{\boldsymbol{j}}}$. 

We can then absorb $(-1)^{r_{\phi_1(\boldsymbol{j})}}$ and $(-1)^{r_{\phi_2(\boldsymbol{j})}}$ onto the definitions of $\gamma_{\phi_1(\boldsymbol{j})}$ and $\gamma_{\phi_2(\boldsymbol{j})}$, respectively. 
We furthermore see that imposing Eq.~(\ref{eq:directioncond}) gives an edge-direction for $\phi(\boldsymbol{j}) \in \widetilde{E}/\widetilde{E}^{\prime}$ that differs from that of the fiducial orientation by the associated sign factor $(-1)^{r_{\phi_1(\boldsymbol{j})} + r_{\phi_2(\boldsymbol{j})}}$, which remains consistent with a redefinition of the signs on the individual Majorana modes. 
Letting $\mathbf{h}$ be the single-particle Hamiltonian for the fiducial orientation, such a redefinition corresponds to conjugating $\mathbf{h}$ by a $\pm 1$ diagonal matrix. 
As no scalar quantity of $\mathbf{h}$ can depend on this choice, this redefinition corresponds to a gauge freedom.

Proceeding this way, we can solve the effective Hamiltonian
\begin{align}
    H_{\boldsymbol{x}} \equiv H \prod_{\{\boldsymbol{j}| \phi(\boldsymbol{j}) \in \widetilde{E}/\widetilde{E}^{\prime}\}} \left(\frac{I + (-1)^{x_{\boldsymbol{j}}} \sigma^{\boldsymbol{y}(\boldsymbol{j})}}{2}\right)
\end{align}

\noindent sector-by-sector over each stabilizer eigenspace designated by $\boldsymbol{x}$. 
If we also need to remove twin vertices from $G(H)$ before it is a line graph, we project onto the mutual $+1$ eigenspace of the stabilizer group $\mathcal{S}_{\boldsymbol{x}}$ defined previously in Section \ref{sec:symmetries} as well. 
Finally, if the parity operator $P$ is trivial in the Pauli description, then only a fixed-parity eigenspace in the fermion description will be physical. 

In the next section, we will see how known free-fermion solutions fit into this characterization and demonstrate how our method can be used to find new free-fermion solvable models, for which we give an example.

\section{Examples}
\label{sec:examples}
\subsection{Small Systems}
\label{sec:small}

The frustration graph of single-qubit Paulis ${X, Y, Z}$ is $K_3$, the complete graph on three vertices. 
This graph is the line graph of not one, but two non-isomorphic graphs: the so-called `claw' graph $K_{1, 3}$, and $K_3$ itself (see Table \ref{tab:notsum}). 
By the Whitney isomorphism theorem \cite{whitney1932congruent}, $K_3$ is the only graph which is not the line graph of a unique graph. 
This ambiguity results in the existence of two distinct free-fermion solutions of a single qubit Hamiltonian, which we will hereafter refer to as ``even" (labeled ``0") and ``odd" (labeled ``1") fermionizations
\begin{align}
    \begin{cases}
    X_0 = i \gamma_0 \gamma_1 & X_1 = i \gamma_2 \gamma_3 \\
    Y_0 = i \gamma_1 \gamma_2 & Y_1 = i \gamma_0 \gamma_3 \\
    Z_0 = i \gamma_0 \gamma_2 & Z_1 = -i \gamma_1 \gamma_3
    \end{cases} \mathrm{.}
    \label{eq:singlequbitferms}
\end{align}
In the even fermionization, no T-join of the root graph $K_3$ exists since there are only three fermion modes $\{\gamma_0, \gamma_1, \gamma_2\}$. 
The orientation of the root graph is constrained by the identity $XYZ = iI$. 
In the odd fermionization, there are four fermion modes $\{\gamma_0, \gamma_1, \gamma_2, \gamma_3\}$, and so a T-join does exist for the root graph $K_{1, 3}$. 
It is the set of all edges of this graph. 
The parity operator is trivial in the Pauli description however, and so the constraint $XYZ = iI$ is enforced by restricting to the $+1$ eigenspace of $P \equiv -\gamma_0 \gamma_1 \gamma_2 \gamma_3$ in the fermion description. 
We are free to choose the orientation of the root graph $K_{1, 3}$ however we like in this case, since it contains no cycles, though this choice will affect what we call the physical eigenspace of $P$. 
By virtue of the line-graph construction and our choice of orientation, both fermionizations respect the single-qubit Pauli multiplication relations, up to stabilizer equivalencies in some cases.

We have made the choice to label the Paulis in the two fermionizations in a compatible way, such that 
\begin{align}
    \sigma^{j}_0 = P \sigma^{j}_1
\end{align}
This gives
\begin{align}
    [\sigma^{j}_{m}, \sigma^{k}_{m}] = 2i \varepsilon_{jk \ell} \sigma^{\ell}_{0} \mathrm{.}
    \label{eq:mplusm}
\end{align}
where $m \in \{0, 1\}$ and $\varepsilon$ is the Levi-Civita tensor. 
Since an even number of parity-operator factors appear on the left side of Eq.~(\ref{eq:mplusm}), the commutator between two Paulis in either fermionization is always a Pauli in the even fermionization. 
We can additionally write multiplication relations between the two fermionizations concisely, as
\begin{align}
    \sigma^{j}_{p} \sigma^{k}_{q} = \delta_{jk} P^{p \oplus q} + i (1 - \delta_{jk}) \varepsilon_{jk \ell} \sigma_{p \oplus q}^{\ell} \mathrm{.}
\end{align}

Another exceptional situation arises for 2 qubits, for which the full frustration graph is the line graph of $K_6$, again depicted graphically in Table~\ref{tab:notsum}. 
This reveals a free-fermion solution for all 2-qubit Hamiltonians by six fermion modes, listed explicitly in Table~\ref{tab:2qubit}. 
We again choose our orientation by picking a spanning tree of the root graph (for example all terms containing the mode $\gamma_5$), choosing an arbitrary orientation on this tree, and choosing the remaining orientations by enforcing the condition in Eq.~(\ref{eq:directioncond}). 
Since $K_6$ has an even number of vertices, there exists a T-join for this graph, e.g.\ the terms $\{XX, YY, ZZ\}$. 
The associated parity operator is trivial however, as $(XX)(YY)(ZZ) = -I$, and so only the $+1$ eigenspace of $P$ in the fermion description will be physical. 
This solution reflects the exceptional Lie algebra isomorphism, $\mathfrak{su}(4) \simeq \mathfrak{spin}(6)$.

Finally, we give an example of a three-qubit Hamiltonian with an exceptional \emph{symmetry}, namely

\begin{align}
    H = XII + YII + ZXX + ZZZ
\end{align}

\noindent This Hamiltonian has the frustration graph shown in the top right entry of Table \ref{tab:strongexcept}, and thus is an exceptional case to Corollary \ref{cor:clifsym}. A symmetry transformation exchanging $e$ and $e^{\prime}$ for this Hamiltonian is the Hadamard gate applied to the second and third qubits, which exchanges the third and fourth terms, but cannot be realized as any permutation of the individual Majorana modes in its free-fermion description.

\begin{figure*}[t!]
    \centering
    \includegraphics[width=0.9\textwidth]{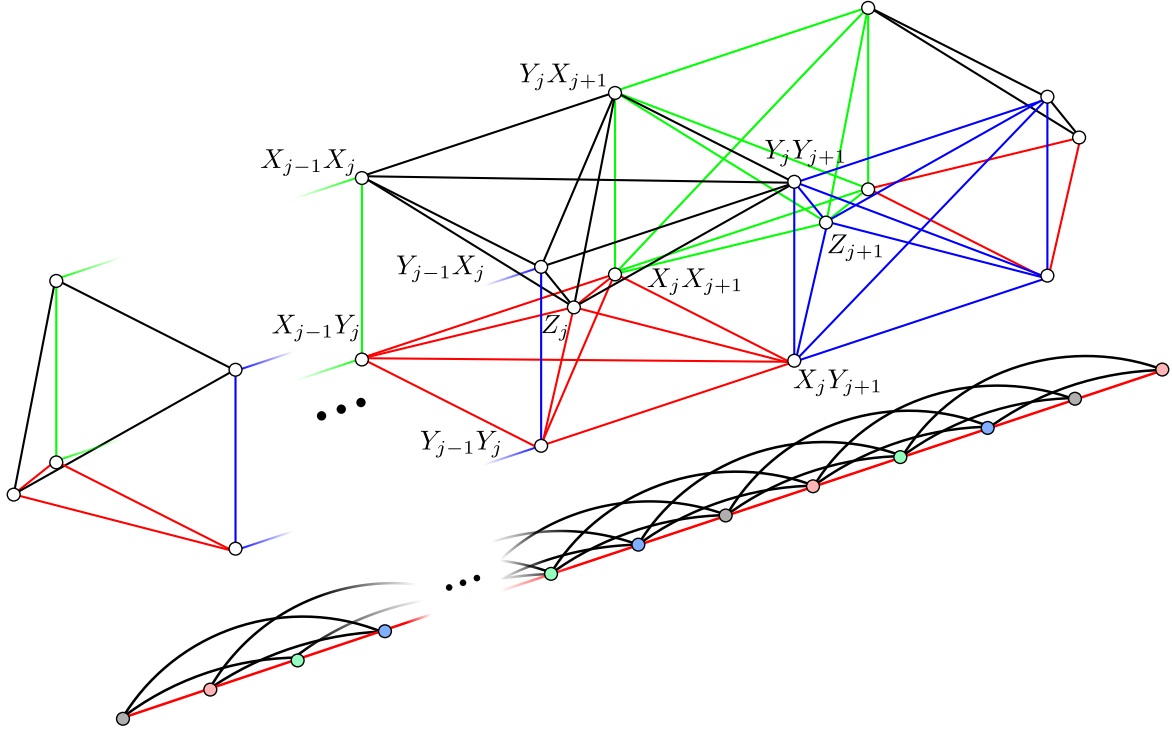}
    \caption{Frustration graph for the general XY model and its root graph, shown below. 
    Cliques are colored to show the Krausz decomposition, which is the image of the model under the Jordan-Wigner transform. 
    Vertices in the root graph are correspondingly colored, and a spanning tree is highlighted.}
    \label{fig:xymodel}
\end{figure*}

\subsection{1-dimensional chains}

Shown in Figure~\ref{fig:xymodel} is the frustration graph $G(H)$ for the most general nearest-neighbor Pauli Hamiltonian in 1-d (on open boundary conditions) which is mapped to a free-fermion Hamiltonian under the Jordan-Wigner transformation,
\begin{align}
H = \sum_{j = 1}^{n - 1} \sum_{\alpha, \beta \in \{x, y\}} \mu^{j}_{\alpha \beta} \sigma_j^{\alpha} \otimes \sigma_{j + 1}^{\beta} + \sum_{j = 1}^n \nu_j Z_j.
\label{eq:1dchain}
\end{align}
Cliques are colored according to the Krausz decomposition of this graph, which is easily seen by the free-fermion description. 
The fermion hopping graph, $R$, is shown below. 
Note that the cycle symmetry subgroup $Z_H$ for this model is trivial, as every product of Hamiltonian terms along a cycle in $R$ is the identity. 
Since the number of vertices in $R$ is even, a T-join does exist, and the parity operator is in-fact $P = Z^{\otimes n}$. 
Therefore, we have $|\mathcal{Z}(\mathcal{P}_H)| = 1$. 

An example spanning tree for the root graph is highlighted, taken simply to be the path along edges $(j, j + 1)$ from $\gamma_1$ to $\gamma_{2n}$. 
Including any additional edge in this tree will form a cycle. 
A natural orientation for this tree is to direct every edge from vertex $j + 1$ to vertex $j$. 
Note that we can recover the Jordan-Wigner transformation from this graphical description alone. 
We first adjoin a single fictitious qubit and a single coupling term to the Hamiltonian, as
\begin{align}
    H^{\prime} = \mu_{xx}^{0} X_0 X_1 + H.
\end{align}
Since the remaining qubits only couple to qubit ``0" along the $X$-direction, all operators in $\mathcal{P}_H$ commute on this qubit. 
Furthermore, this new term adds one vertex to the black clique at the left boundary of the chain in Fig.~\ref{fig:xymodel}. 
It thus extends the spanning tree of $R$ by one vertex -- which we label $\gamma_0$ -- due to the fact that this new term only belongs to one clique (so we take its additional Majorana mode to be a clique of size zero). 
It can be easily verified that products of Hamiltonian terms from this new vertex to any vertex along the chosen spanning tree have the form
\begin{align}
    \begin{cases}
        i \gamma_{0} \gamma_{2 j - 1} \equiv X_0 \bigotimes_{k = 1}^{j - 1} Z_k \otimes X_j & \\
        i \gamma_{0} \gamma_{2j} \equiv  X_0 \bigotimes_{k = 1}^{j - 1} Z_k \otimes Y_j & \\
    \end{cases}
\end{align}
All such operators share $\gamma_0$, so their commutation relations are unchanged by truncating $\gamma_0$. 
Furthermore, since all operators in $\mathcal{P}_H$ commute on qubit-0, we may truncate this qubit as well without changing the commutation relations of the operators above to obtain the Jordan-Wigner transformation
\begin{align}
    \begin{cases}
        \gamma_{2 j - 1} \equiv \bigotimes_{k = 1}^{j - 1} Z_k \otimes X_j & j \ \mathrm{odd} \\
        \gamma_{2j} \equiv \bigotimes_{k = 1}^{j - 1} Z_k \otimes Y_j & j \ \mathrm{odd} \\
    \end{cases}
\end{align}
In principle, a similar trick would work in general, but we find it generally simpler to define Majorana quadratic operators to avoid truncating at a boundary. Our method is especially convenient when considering the case of periodic boundary conditions on this model, wherein we add the boundary term

\begin{align}
H_{\mathrm{boundary}} = \sum_{\alpha, \beta \in \{x, y\}} \mu^{n}_{\alpha \beta} \sigma_1^{\alpha} \otimes \sigma_{n}^{\beta} 
\end{align}
to the Hamiltonian in Eq.~(\ref{eq:1dchain}). With this term included, the model has a nontrivial cycle symmetry given by taking the product of fermion bilinears around the periodic boundary, and this product is also proportional to the parity operator $P = Z^{\otimes n}$ in the spin picture. Adding a boundary term therefore does not change $|\mathcal{Z}(\mathcal{P}_H)|$, though it does require that we solve the model over each of the eigenspaces of the cycle symmetry independently by choosing the sign of the additional terms in the fermion picture as described in Section~\ref{sec:orientation}. Let the eigenspace of $Z^{\otimes n}$ be specified by the eigenvalue $(-1)^p$. For each associated free-fermion model solution, we must then restrict to the $+1$ eigenspace of the parity operator in the spin picture

\begin{align}
    \prod_{j = 1}^n X_j X_{j + 1} \mapsto (-i)^n (-1)^p \prod_{k = 1}^{2n} \gamma_k \mathrm{.}
\end{align}
where index addition is taken modulo $n$. This ensures that our free-fermionic solution respects the constraint

\begin{figure*}[t!]
    \centering
    \includegraphics[width=\textwidth]{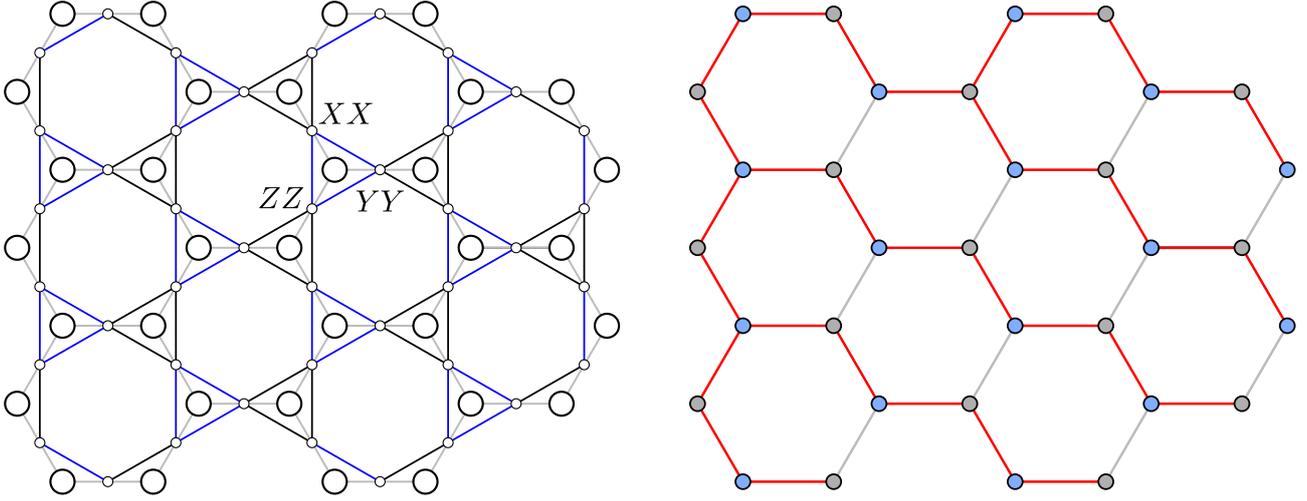}
    \caption{Frustration graph for the Kitaev honeycomb model (left) and its root graph (right). 
    Cliques are colored to show the Krausz decomposition. 
    Interestingly, this model's root graph is the same as its interaction graph. 
    A spanning tree of the root is again highlighted.}
    \label{fig:hcmodel}
\end{figure*}

\begin{align}
    \prod_{j = 1}^n X_j X_{j + 1} = I\mathrm{.}
\end{align}
Notice that solving the two free fermion models together (one for each eigenspace of $Z^{\otimes n}$) gives $2^{n + 1}$ eigenstates, yet restricting to a fixed-parity sector in each keeps only $2^n$ of them, as required. Finally, we see that this model contains no logical qubits via
\begin{align}
    n_L = n - \left[\frac{1}{2}(2n - 2) + 1\right] = 0
\end{align}
as we might expect. 

\subsection{The Kitaev honeycomb model}

Next we consider the Kitaev honeycomb model in two dimensions \cite{kitaev2006anyons}. 
This model has the Hamiltonian
\begin{align}
    H = \sum_{\alpha \in \{x, y, z\}} \sum_{\alpha-\mathrm{links} \ j} J^{j}_{\alpha} \sigma^{\alpha}_j \sigma^{\alpha}_{j + \hat{\alpha}} 
\end{align}
where each of the $\alpha$ links correspond to one of the compass directions of the edges of a honeycomb lattice. 
Once again, the frustration graph with shaded cliques according to the Krausz decomposition is shown in Fig.~\ref{fig:hcmodel}. 
Interestingly, the root graph of this model's frustration graph is again the honeycomb lattice. 
By going backwards, we can see that indeed, any free-fermion model with trivalent hopping graph can be embedded in a 2-body qubit Hamiltonian with the same interaction graph. 
This is because we can find a set of Pauli operators satisfying any frustration graph whose edges can be partitioned into triangles by assigning a different single-qubit Pauli to each of the vertices of every triangle. 
A term in the Hamiltonian is then the tensor product of all of the Pauli operators from the triangles to which its vertex in $G(H)$ belongs. 

Unlike in the one-dimensional example, the cycle subgroup of this model, $Z_H$, is nontrivial. 
This subgroup is generated by the products of Hamiltonian terms around a hexagonal plaquette of the honeycomb lattice, denoted $W_p$ for plaquette $p$. 
These cycles are not independent, however, with constraints between them depending on the boundary conditions of the lattice. 
In particular, if the model is on a torus of dimension $L_x$ by $L_y$, then the product of all Hamiltonian terms is trivial
\begin{align}
    \prod_{\boldsymbol{j} \in V} \sigma^{\boldsymbol{j}} = (-1)^{L_x L_y} I
\end{align}
In this case, the cycles of the honeycomb lattice are not independent, since they similarly multiply to the identity. 
There are thus $L_x L_y - 1$ independent plaquettes on the lattice. 
There are additionally two homotopically nontrivial cycles, which are independent as well. 
Notice that the edges of the honeycomb lattice itself form a T-join, and so the above constraint is also the statement that $P$ is furthermore trivial. 
Therefore, we have
\begin{align}
    |\mathcal{Z}(\mathcal{P}_H)| = L_x L_y - 1 + 2 = L_x L_y + 1
\end{align}
and once again (as first computed in Ref.~\cite{Suchara2011})
\begin{align}
    n_L =2 L_x L_y - \left[\frac{1}{2} \left(2 L_x L_y - 2\right) + L_x L_y + 1 \right] = 0.
\end{align}
This example also illustrates that quite a large number of symmetries could be present, and in general this will complicate finding, e.g., the symmetry sector that contains the ground state.

\begin{figure*}[t!]
    \centering
    \includegraphics[width=\textwidth]{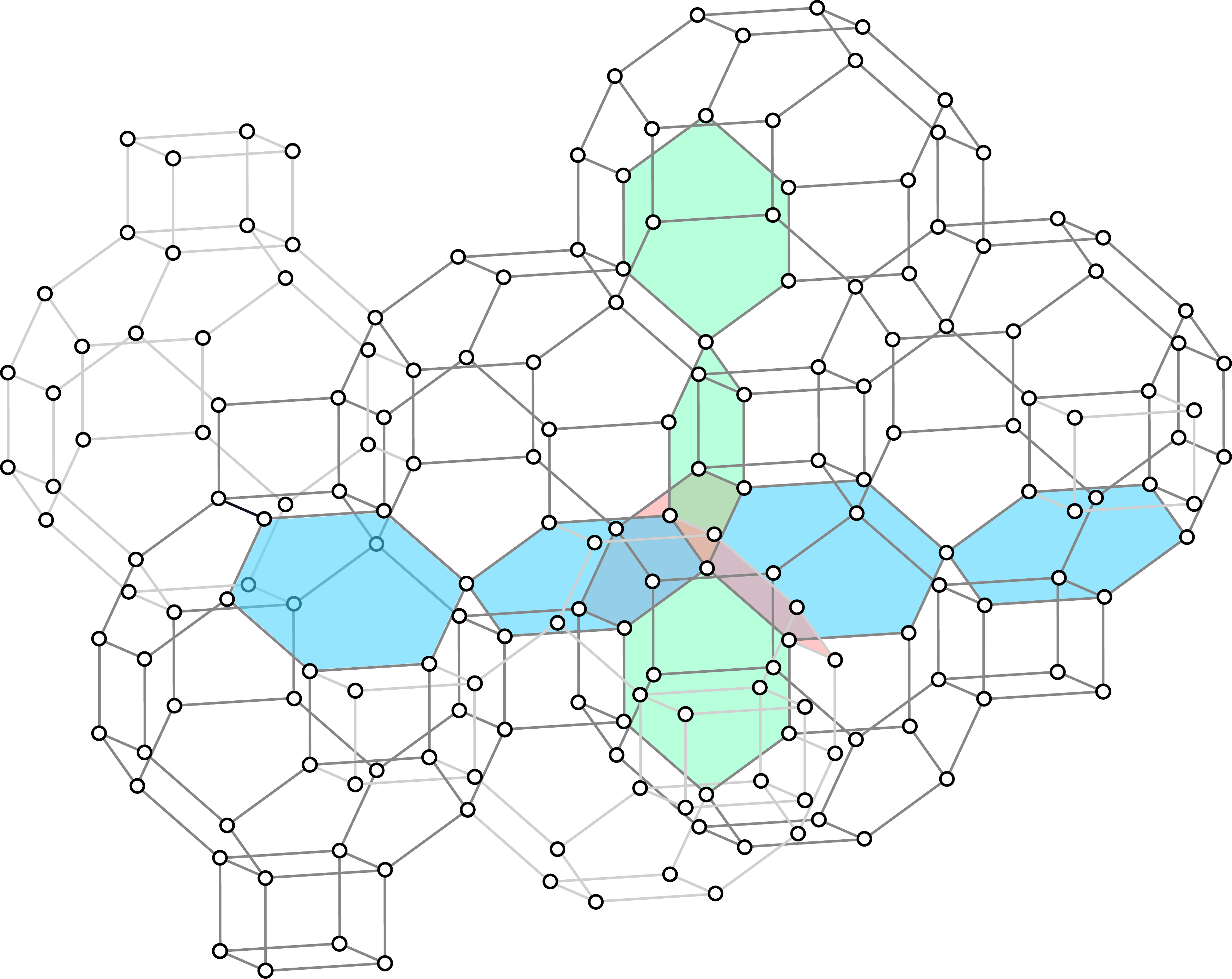}
    \caption{The frustrated hexagonal gauge 3d color code, proposed in Ref.~\cite{roberts2019symmetry}. This model is based on the 3d gauge color code, whose qubits live on the vertices of the lattice shown. Gauge generators for the 3d gauge color code consist of Pauli-$Z$ and Pauli-$X$ operators around both the square and hexagonal faces of the lattice. Stabilizers of the 3d gauge color code consist of Pauli-$Z$ and Pauli-$X$ operators on both the cube and ``ball" cells. The frustrated hexagonal gauge 3D color code is given by taking the stabilizers of the gauge color code together with the hexagonal gauge generators, which commute with the stabilizers, but not with each other. We see from the colored hexagonal faces above that the frustration graph of these gauge generators is a set of disconnected path graphs. Every hexagonal plaquette term anticommutes with exactly two others---the plaquette terms of the other Pauli type intersecting it at exactly one qubit---and commutes with all other terms in the Hamiltonian. }
    \label{fig:gcc}
\end{figure*}

\subsection{Frustrated
Hexagonal Gauge 3D Color Code}

\begin{figure*}[t!]
    \centering
    \includegraphics[width=\textwidth]{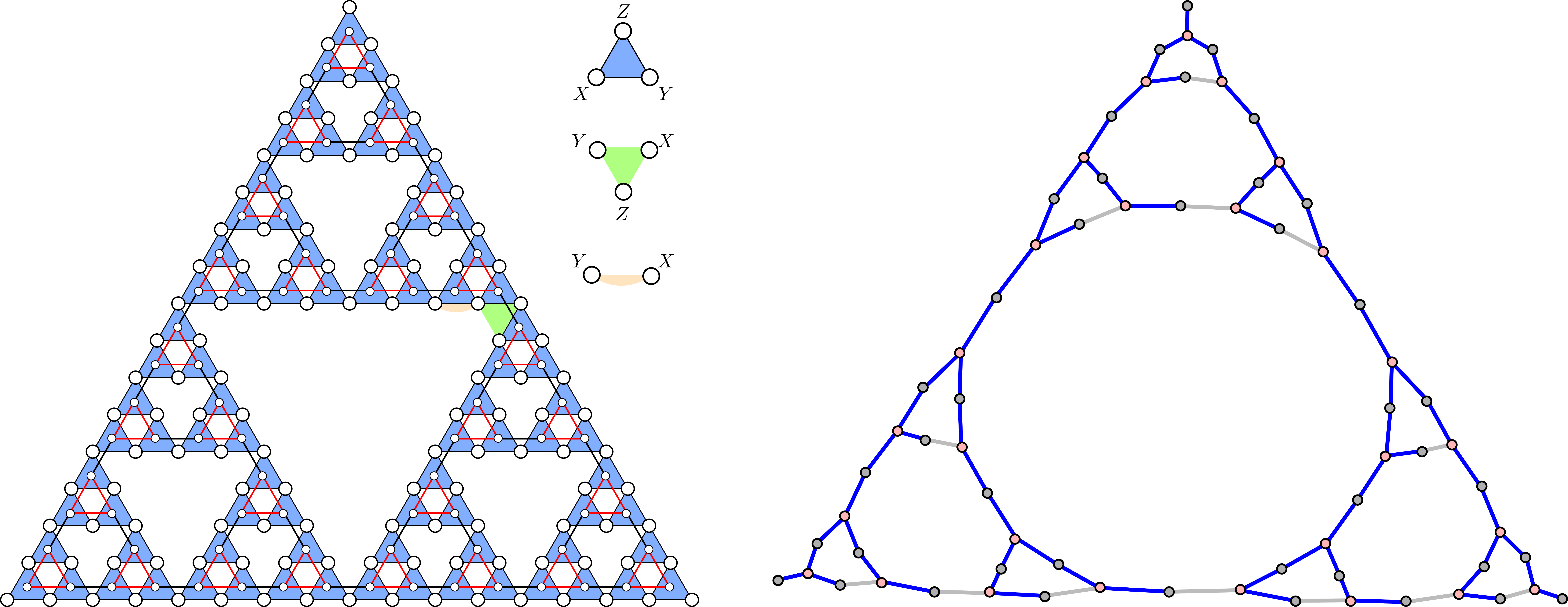}
    \caption{The Sierpinski-Hanoi model (left) with its frustration graph, highlighted, and its root graph (right) with a spanning tree highlighted, for $k = 5$ and local fields absent. 
    Hamiltonian terms are 3-qubit operators acting on qubits at the vertices of the Sierpinski sieve graph, highlighted in blue. 
    Cliques of the frustration graph are colored to show the graph's Krausz decomposition. 
    Green and orange cells depict generators for the model's logical Pauli group. 
    At the interior triangular cells of the lattice are the 3-body generators shown in green. 
    At the interior and exterior edges of the model are 2-body generators shown in orange. 
    These are obtained from their adjoining Hamiltonian terms by reflecting the action on the intersection of their supports (so these generators act differently depending on which edge they act on). 
    The frustration graph of this model is the Hanoi graph $H_3^{k - 1}$. 
    The vertices of this graph are in correspondence to the states of the towers of Hanoi problem with three towers and $k - 1$ discs. 
    The root graph of $H_3^{k - 1}$ contains $H_3^{k - 2}$ as a topological minor.}
    \label{fig:shmodel}
\end{figure*}

The frustrated hexagonal gauge 3d color code is a noncommuting Hamiltonian whose terms consist of the stabilizer generators and a subset of the gauge generators from the gauge color code. The gauge color code \cite{bombin2016dimensional, bombin2015singleshot, kubica2015universal, bombin2015gauge} has a Hamiltonian that is defined in terms of a natural set of gauge generators as
\begin{align}
    H = -\sum_{S \in \square, \varhexagon} \left(J_x \bigotimes_{j \in S} X_j + J_Z \bigotimes_{j \in S} Z_j \right) \nonumber \\
    + \ \text{(boundary terms)}
    \label{eq:gcchamiltonian}
\end{align}
where ``$\square$" and ``$\varhexagon$" denote the sets of square and hexagonal faces on the lattice in Fig.~\ref{fig:gcc}, respectively (see Ref.~\cite{brown2016faulttolerant} for a detailed description of this lattice). 
Here the qubits live on the vertices of the lattice. 
Nontrivial boundary conditions are required to restrict the logical space of this code to a single qubit. 
We will ignore these boundary conditions and consider only gauge generators in the bulk of the lattice. 
The stabilizers for this model are given by products of $X$ or $Z$ around every elementary cell, either a cube or a ``ball". 

We consider a model where we partially restore some of these symmetries. 
Namely, we will consider the cube and balls to be ``restored'' symmetries of the model, and we will remove the square generators.
This leads to the following gauge Hamiltonian that sums over only hexagonal faces, balls, and cubes,
\begin{align}
    H = -\sum_{S \in \varhexagon, \text{\mancube}, \myfancysymbol} \left(J_X \bigotimes_{j \in S} X_j + J_Z \bigotimes_{j \in S} Z_j \right).
    \label{eq:gccrestored}
\end{align}
The cube and ball terms commute with all of the hexagon terms, and so constitute symmetries of the model. 
Once we fix a sector for these terms, we can solve the remaining model by mapping to free-fermions as follows~\cite{roberts2019symmetry}. 

In Figure~\ref{fig:gcc}, we represent a subsection of the qubit lattice of this code, where qubits live at the tetravalent vertices. 
Because the cube and ball terms commute with everything, the frustration graph depends only on the hexagonal faces, several of which are colored in Figure~\ref{fig:gcc}. 
We see that some of these faces intersect at exactly one vertex, and so the $X$- and $Z$-type gauge generators will anticommute on the associated qubit. 
These intersection patterns only occur in 1D chains along the cardinal axes of the lattice. 
In particular, every hexagonal face only overlaps with two other hexagonal faces along these chains and otherwise intersects the other faces at an even number of qubits. 
The frustration graph of this model thus decouples into a set of disconnected paths, which are line graphs, and in fact they are the frustration graph of the XY-model~\cite{lieb1961two} and the 1-d Kitaev wire \cite{kitaev2001unpaired}. 
A free-fermion mapping therefore exists for this model, and this demonstrates an example of how one might construct subsystem codes with a free-fermion solution to obtain desired spectral properties.
In particular, when $|J_x| \not= |J_z|$ the model in Eq.~(\ref{eq:gccrestored}) is gapped. 
We note that this observation was made previously in Ref.~\cite{roberts2019symmetry} in the context of quantum error correcting codes. 

\begin{figure}[t!]
    \centering
    \includegraphics[width=0.5\textwidth]{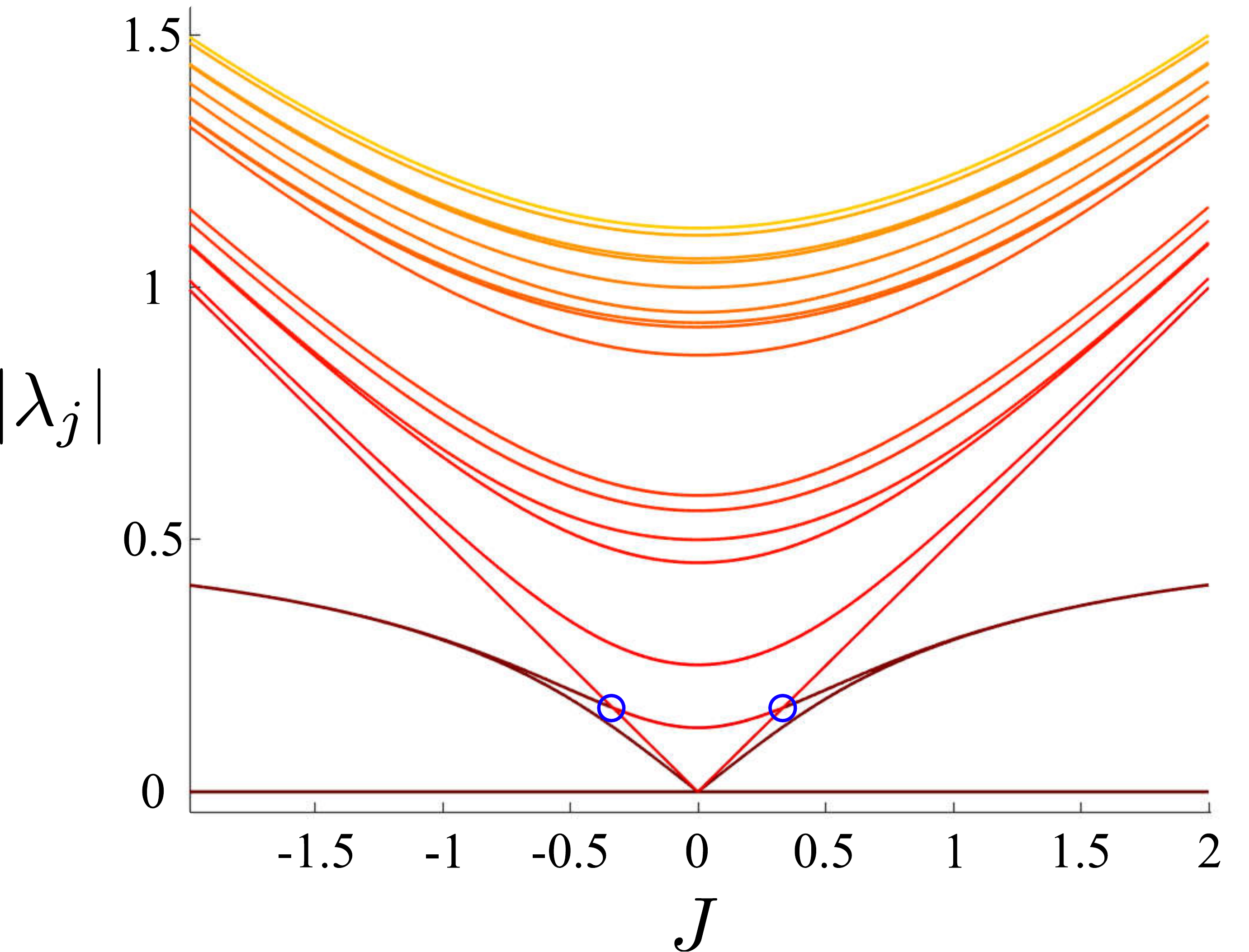}
    \caption{Single-Particle spectrum of the Sierpinski-Hanoi model for $k = 5$ with an additional local field term present in the symmetry sector for which all cycles are $+1$. 
    Circled are two critical points where excited bands become degenerate.}
    \label{fig:shmodelspectrum}
\end{figure}

\subsection{Sierpinski-Hanoi model}

Finally, we introduce our own example of a solvable spin model, which was previously unknown to the best of our knowledge. 
This model consists of 3-body $XYZ$-interaction terms on the shaded cells of the Sierpinski triangle, all with the same orientation, as depicted in Fig.~\ref{fig:shmodel}.  Explicitly,  the Hamiltonian for this model is given by

\begin{align}
    H = \sum_{(i, j, k) \in \color{babyblueeyes}\mbox{\normalsize$\blacktriangle$}} X_i Y_j Z_k + J H_{\text{local}}
\end{align}
where $(i, j, k)$ is an ordered triple of qubits belonging to a particular shaded cell on the lattice and we will define additional on-site terms $H_{\text{local}}$ in Eq.~(\ref{eq:shlocaldef}). 
An instance of the model is parameterized by $k$, the fractal recursion depth of the underlying Sierpinski lattice, where $k = 1$ is taken to be a single 3-qubit interaction. 

Let us first consider the simplified model where $J=0$.
Then the frustration graph of this model is the so-called \emph{Hanoi graph} $H_{3}^{k - 1}$. 
The vertices of this graph are labeled by states of the towers of Hanoi problem with $k - 1$ discs, and two vertices are neighboring if transitioning between the corresponding states is an allowed move in the problem. 
Perhaps surprisingly, this graph is a line graph, with root and highlighted spanning tree shown in Fig.~\ref{fig:shmodel}. 
Furthermore, the root graph contains $H_{3}^{k - 2}$ as a topological minor, obtained by removing the vertices of degree one and contracting the vertices of degree two each along one of their two edges.

This model contains 
\begin{align}
    n = \frac{3}{2} (3^{k - 1} + 1)
\end{align}
physical qubits, and its root graph contains
\begin{align}
    |\widetilde{V}| = 
    \begin{cases} 
        2 & k = 1 \\
        \frac{1}{2}\left[5 \times 3^{k - 2} + 3\right] & k > 1
    \end{cases}
\end{align}
vertices. 
A T-join for this graph therefore exists only for even $k$ and $k = 1$, and the parity operator is never trivial when a T-join exists. 
The root graph also contains
\begin{align}
    |Z_H| = \sum_{j = 0}^{k - 3} 3^j = 
    \begin{cases}
        0 & k \leq 2 \\
        \frac{1}{2} \left(3^{k - 2} - 1\right) & k > 2
    \end{cases}
\end{align}
fundamental cycles. 
None of the generators of the cycle subgroup are trivial since every qubit is acted upon by at most two anti-commuting operators. 
The number of logical qubits in this model is therefore
\begin{align}
    n_{L} = 
    \begin{cases}
        2 & k = 1 \\
        \frac{1}{4} \left[11\times 3^{k - 2} + 8+(-1)^k\right] & k > 1,
    \end{cases}
\end{align}
and so this Hamiltonian encodes logical qubits at a constant rate of $\frac{11}{18}$ in the infinite $k$ limit. 
Perhaps unsurprisingly, these logical qubits live on the boundaries of the fractal, and we can obtain a set of generators for the logical Pauli group of this model as shown in Fig.~\ref{fig:shmodel}. 
We can encode logical quantum information in this model by picking symplectic pairs of generators from this group, which anticommute with one another yet commute with the remaining generators in the group. 
The remaining such generators can be used as gauge qubits for the logical qubit we wish to protect, and the free-fermion Hamiltonian of the model can be used for error suppression. 

We are also free to add an anisotropic local-field term to a subset of the qubits without breaking solvability
\begin{align}
    H_{\mathrm{local}} = \sum_{i \in \begingroup\color{babyblueeyes}\mbox{\normalsize$\blacktriangle$}\endgroup \mathbf{-} \begingroup\color{babyblueeyes}\mbox{\normalsize$\blacktriangle$}\endgroup} \sigma_i^{j_i}
    \label{eq:shlocaldef}
\end{align}
The sum is taken over all qubits corresponding to \emph{black} edges connecting two shaded cells in Fig.~\ref{fig:shmodel}. 
$j_i$ is the third Pauli type from the two Paulis acting on qubit $i$ by the interaction terms. 
The effect of these local-field terms is to couple every black vertex in the root graph shown in Fig.~\ref{fig:shmodel}, except for those at the three corners, to a dedicated fermion mode. 
We do not depict these additional modes to avoid cluttering the figure. 
These terms also do not affect the symmetries of the model, except possibly to add the parity operator to $\mathcal{Z}(\mathcal{P}_H)$ when the number of vertices in the original graph was odd, as the number of vertices in the graph with local field terms present will always be even. 
The parity operator can then be constructed as the product of all of the Hamiltonian terms, since the root graph only has vertices of degrees 1 and 3.

In Fig.~\ref{fig:shmodelspectrum}, we display the single-particle spectrum of the Sierpinski-Hanoi model as a function of the local field $J$ for $k = 5$ in the sector for which all of the cycle symmetries are in their mutual $+1$ eigenspace. We highlight two critical points where excited energy levels become degenerate to within our numerical precision. We observe that the locations of these points are not system-size-independent, but rather asymptotically approach $J = 0$ as the system size is increased. We conjecture that this is connected to the emergence of scale symmetry, which the model possesses in the thermodynamic limit, yet not for any finite size. It would be intriguing if certain physical features of this symmetry could be realized at the critical points at finite size, potentially opening the door to simulating scale-invariant systems on a finite-sized quantum computer.

\section{Proofs of Main Theorems}

\subsection{Proof of Theorem 1}
\label{sec:thm1proof}
We restate Theorem \ref{thm:ffsolution} for convenience. 

\begin{customthm}{1, restated}[Existence of free-fermion solution] 
An injective map $\phi$ as defined in Eq.~(\ref{eq:paulitopair}) and Eq.~(\ref{eq:scommeq}) exists for the Hamiltonian $H$ as defined in Eq.~(\ref{eq:hdef}) if and only if there exists a root graph $R$ such that
\begin{align}
    G(H) \simeq L(R),
\label{eq:ffsolutionrestatement}
\end{align} 
where R is the hopping graph of the free-fermion solution.
\label{thm:ffsolutionrestatement}
\end{customthm}

\begin{proof}
If $\phi$ exists, define $R = (V, E)$, where $E \equiv \{(\phi_1(\boldsymbol{j}), \phi_2(\boldsymbol{j})) | \phi_1(\boldsymbol{j}), \phi_2(\boldsymbol{j}) \in V, \boldsymbol{j} \in E\}$. 
If and only if $|(\phi_1(\boldsymbol{j}), \phi_2(\boldsymbol{j})) \cap (\phi_1(\boldsymbol{k}), \phi_2(\boldsymbol{k}))| = 1$, then the vertices corresponding to $\boldsymbol{j}$ and $\boldsymbol{k}$ are neighboring in $G$ by Eq.~(\ref{eq:scommeq}). 
Thus, $G(H) \simeq L(R)$ and a mapping $\phi$ exists only if $R$ does.

If there exists a graph $R \equiv (V, E)$ such that $G \simeq L(R)$, take the Krausz decomposition of $G(H)$. 
Namely, partition the edges of $G(H)$ as $F = \{C_1, \dots, C_{|V|}\}$, where each $C_i$ constitutes a clique in $G$ and such that every vertex in $G$ appears in at most two $C_i$. 
The cliques in this partitioning correspond to the vertices $V$ of $R$. 
For each vertex $\boldsymbol{j}$, define $\phi(\boldsymbol{j})$ to be the pair of cliques in which $\boldsymbol{j}$ appears. 
Since the cliques partition the edges of $G$, then if vertices $\boldsymbol{j}$ and $\boldsymbol{k}$ are neighboring in $G$, they must appear in exactly one clique together, and thus $|(\phi_1(\boldsymbol{j}), \phi_2(\boldsymbol{j})) \cap (\phi_1(\boldsymbol{k}), \phi_2(\boldsymbol{k}))| = 1$. 
Thus, $\phi$ satisfies Eq.~(\ref{eq:scommeq}). 
Furthermore, $\phi$ is injective, since if there are two vertices $\boldsymbol{j}$, $\boldsymbol{k} \in G$ such that $\phi(\boldsymbol{j}) = \phi(\boldsymbol{k})$, then $\boldsymbol{j}$ and $\boldsymbol{k}$ appear in the same two cliques, but since the Krausz decomposition is a partition of the edges, this would require that $\boldsymbol{j}$ and $\boldsymbol{k}$ neighbor by two edges. 
However, the definition of $G$ guarantees that pairs of vertices can only neighbor by at most one edge, and so this is impossible. 
Therefore $\phi$ is injective. 
\end{proof}

\subsection{Proof of Theorem 2}
\label{sec:thm2proof}

Once again, we restate our theorem for convenience

\begin{customthm}{2, restated}[Symmetries are cycles and parity]
Given a Hamiltonian satisfying Eq.~(\ref{eq:ffsolutionrestatement}) such that the number of vertices $|\widetilde{V}|$ in the root graph is odd, then we have
\begin{align}
    \mathcal{Z}(\mathcal{P}_H) = Z_{H}.
\end{align}
\noindent If the number of vertices in the root graph is even, then we have
\begin{align}
\mathcal{Z}(\mathcal{P}_H) = \left\langle Z_{H}, P \right\rangle.
\end{align}
\label{thm:symmetriesrestated}
\end{customthm}

\emph{Proof.} Let $G \equiv (E, F) \simeq L(R)$ be the connected line graph of a connected root graph $R = (V, E)$, and let $G$ have adjacency matrix $\mathbf{A}$. 
We will need the following well-known factorization of a line graph adjacency matrix $\mathbf{A}$
\begin{align}
    \mathbf{A} = \mathbf{B} \mathbf{B}^{\mathrm{T}} \ \mathrm{(mod \ 2)}
\end{align}
where $\mathbf{B}$ is the edge-vertex incidence matrix of $R$. 
That is, $\mathbf{B}$ is a $|E| \times |V|$ matrix such that
\begin{align}
    B_{\boldsymbol{j} l} = 
    \begin{cases}
        1 & l \in \boldsymbol{j} \\
        0 & \mathrm{otherwise}
    \end{cases}
\end{align}
for all $\boldsymbol{j} \in E$ and $l \in V$. 
We can interpret $\mathbf{B}$ as defining the map $\phi$ via
\begin{align}
    \phi : \sigma^{\boldsymbol{j}} \mapsto \prod_{l \in V} \gamma_l^{B_{\boldsymbol{j} l}} 
\end{align}
That is, $\phi_1(\boldsymbol{j})$ and $\phi_2(\boldsymbol{j})$ are the indices of the nonzero elements in the row labeled by $\boldsymbol{j}$ in $\mathbf{B}$. 
This then defines the adjacency matrix $\mathbf{A}$ through the scalar commutator as
\begin{align}
    \scomm{\prod_{l \in V} \gamma_l^{B_{\boldsymbol{j} l}}}{\prod_{m \in V} \gamma_m^{B_{\boldsymbol{k} l}}} &= \prod_{l, m \in V} \scomm{\gamma_l^{B_{\boldsymbol{j} l}}} {\gamma_m^{B_{\boldsymbol{k} m}}} \\
    &= \prod_{l, m \in V} (-1)^{(1 - \delta_{lm}) B_{\boldsymbol{j} l} B_{\boldsymbol{k} m}} \\
    \scomm{\prod_{l \in V} \gamma_l^{B_{\boldsymbol{j} l}}}{\prod_{m \in V} \gamma_m^{B_{\boldsymbol{k} l}}} &= (-1)^{(\mathbf{B} \mathbf{B}^{\mathrm{T}})_{\boldsymbol{j} \boldsymbol{k}} + \left(\sum_{l} B_{\boldsymbol{j} l}\right) \left(\sum_{l} B_{\boldsymbol{k} l}\right)} \\
    (-1)^{A_{\boldsymbol{j} \boldsymbol{k}}} &= (-1)^{(\mathbf{B} \mathbf{B}^{\mathrm{T}})_{\boldsymbol{j} \boldsymbol{k}}}
\end{align}
From the third to the fourth line, we replaced the left-hand side with the definition of $\mathbf{A}$ and used the fact that the rows of $\mathbf{B}$ have exactly two nonzero elements. 
By the distributive property of the scalar commutator Eq.~(\ref{eq:distribution}), we can extend the above equation to products of Hamiltonian terms
\begin{align}
    \prod_{\boldsymbol{j} \in E} \left(\prod_{l \in V} \gamma_l^{B_{\boldsymbol{j} l}}\right)^{v_{\boldsymbol{j}}} = \pm \prod_{l \in V} \gamma_l^{\left(\mathbf{B}^{\mathrm{T}} \cdot \mathbf{v} \right)_{l}} \mathrm{,}
\end{align}
where $\mathbf{v} \in \{0, 1\}^{\times |E|}$, as
\begin{align}
    \scomm{\prod_{l \in V} \gamma_l^{B_{\boldsymbol{j} l}}}{\prod_{\boldsymbol{k} \in E} \left(\prod_{m \in V} \gamma_m^{B_{\boldsymbol{k} m}}\right)^{v_{\boldsymbol{k}}}} = (-1)^{\left(\mathbf{B}\mathbf{B}^{\mathrm{T}} \cdot \mathbf{v} \right)_{\boldsymbol{j}}} 
\end{align}
since linear combinations of rows of $\mathbf{B}$ over $\mathds{F}_2$ will have even-many ones. 
Every element of $\mathcal{P}_H$ is a (non-unique) linear combination of rows of $\mathbf{B}$ over $\mathds{F}_2$, and so to characterize the elements of $\mathcal{Z}(\mathcal{P}_H)$, it is sufficient to find a spanning set of the kernel of $\mathbf{A}$,
\begin{align}
    \mathbf{A} \cdot \mathbf{v} = \mathbf{B} \mathbf{B}^{\mathrm{T}} \cdot \mathbf{v} = \mathbf{0} \ \mathrm{(mod \ 2)}
\end{align}
It is again well-known that the $\mathds{F}_2$-kernel of $\mathbf{B}^{\mathrm{T}}$ is the cycle space of $R$, and this specifies the cycle subgroup $Z_H$ as being contained in $\mathcal{Z}(\mathcal{P}_H)$. 
All that is left is therefore to find all $\mathbf{v}$ such that $\mathbf{B}^{\mathrm{T}} \cdot \mathbf{v}$ is in the kernel of $\mathbf{B}$. 
Since we have assumed $G$ is connected, it is easy to see that the only element in this kernel is $\mathbf{1}$, the all-ones vector. 
Thus, $\mathbf{v}$ will also be in the kernel of $\mathbf{A}$ if it defines a T-join of $G$. 
If $|V|$ is even, then we can construct a T-join by first pairing the vertices along paths of $G$. 
We can then ensure that each edge appears at most once in the T-join by taking the symmetric difference of all paths. 
If $|V|$ is odd, then no T-join exists. 
Indeed, assume that a T-join $T$ does exist for $|V|$ odd, and let $\widetilde{G} = (V, T) \subseteq G$ be the subgraph of $G$ containing exactly the edges from the T-join. 
By construction $\widetilde{G}$ contains all the vertices of $G$ and has odd degree for every vertex, though it may no longer be connected. 
Let these degrees be $\{d_j\}_{j \in V}$, then by the handshaking lemma
\begin{align}
    \sum_{j = 1}^{|V|} d_j = 2|T| \mathrm{.}
\end{align}
However, the left side must be odd since we have assumed the degree of every vertex in $\widetilde{G}$ is odd, and the number of vertices is also odd, and so we have a contradiction. \qed

\section{Discussion}

We have seen how the tools of graph theory can be leveraged to solve a wide class of spin models via mapping to free fermions, and given an explicit procedure for constructing the free-fermion solution when one exists. 
A major remaining open question, however, concerns the characterization of free-fermion solutions beyond the generator-to-generator mappings we consider here. 
That is, if $G(H)$ is not a line graph and no removal of twin vertices will make it so, then it may \emph{still} be possible for a free-fermion solution for $H$ to exist thanks to the continuum of locally equivalent Pauli-bases into which $H$ may be expanded. 
Our fundamental theorem does not rule out the possibility that special such bases may exist.
The problem of finding such bases is equivalent to finding specific unitary rotations of $H$ for which the $G(H)$ again becomes a line graph.
These rotations must be outside of the Clifford group, since the frustration graph is a Clifford invariant. 
Their existence may therefore depend on specific algebraic relationships between the Pauli coefficients $h_{\boldsymbol{j}}$ in the Hamiltonian, since the existence of a free-fermionization is a spectral invariant. 
We expect such transformations will be hard to find in general, though perhaps progress can be made for single-qubit rotations on 2-local Hamiltonians in a similar vein as in Ref.~\cite{klassen2019twolocalqubit} for stoquasticity. Recently, a local spin-$\sfrac{1}{2}$ model with a free-fermion solution -- despite no such generator-to-generator solution existing -- has been found in Ref.~\cite{fendley2019free}. An investigation of models which may be fermionized by these more general transformations is therefore an interesting subject of future work.

It is natural to ask whether our results could have implications for simulating quantum systems and quantum computation. 
We expect our characterization to shed some light on the inverse problem of finding fermion-to-qubit mappings, such as the Bravyi-Kitaev superfast encoding \cite{bravyi2002fermionic}, Bravyi-Kitaev transform \cite{seeley2012bravyikitaev}, and generalized superfast encoding \cite{setia2018superfast}. 
It is possible to achieve further encodings by introducing ancillary fermion modes, as seen in the Verstraete-Cirac mapping \cite{verstraete2005mapping} and the contemporaneous mapping introduced by Ball \cite{ball2005fermions}. 
Encodings can be further improved through tailoring to specific symmetries \cite{bravyi2017tapering}, connectivity structures \cite{steudtner2018fermion, jiang2018majorana}, and through the application of Fenwick trees \cite{havlicek2017operator}. 
Recently, a treelike mapping was shown to achieve optimal average-case Pauli-weight in Ref~\cite{jiang2019optimal}. 
In their ``Discussion'' section, the authors remark that an interesting future direction for their work would involve introducing ancillary qubits to their mapping. 
We expect our classification of the symmetries of free-fermion spin models to help guide this investigation, though further work is required to fully characterize the logical symmetry groups which can be realized by these models. 

Finally, our characterization highlights the possibility of a ``free-fermion rank" for Hamiltonians as an important measure of classical simulability. 
Namely, if there is no free-fermion solution for a given Hamiltonian, we can still group terms into collections such that each collection independently has such a solution. 
An interesting natural question for future work is whether the minimal number of such collections required can be interpreted as a quantum resource in an analogous way to the fermionic Gaussian rank \cite{bravyi2017complexity} or stabilizer rank \cite{bravyi2019simulation} for states.

\begin{acknowledgements}
We thank Samuel Elman, Ben Macintosh, Ryan Mann, Nick Menicucci, Andrew Doherty, Stephen Bartlett, Sam Roberts, Alicia Koll\'{a}r, Deniz Stiegemann, Sayonee Ray, Chris Jackson, Jonathan Gross, and Nicholas Rubin for valuable discussions throughout this project. 
This work was supported by the Australian Research Council via EQuS project number CE170100009.
\end{acknowledgements}

\end{document}